\documentclass[journal, 10pt]{IEEEtran}
\usepackage[colorlinks,
            linkcolor=blue,
            anchorcolor=red,
            citecolor=red]{hyperref}

\usepackage{booktabs}

\usepackage{diagbox} 
\usepackage{graphicx}
\usepackage{epstopdf}
\usepackage{psfrag}
\usepackage{subfigure}
\usepackage{url}
\usepackage{stfloats}
\usepackage{amsfonts,amssymb,amsmath,bm,paralist,theorem,cite,ifthen,color,nccmath}
\usepackage{caption}
\usepackage{calc}
\usepackage{enumerate}
\usepackage{multirow}
\usepackage{makecell}
\usepackage[linesnumbered,ruled]{algorithm2e}
\usepackage{setspace}
\usepackage{array}
\usepackage{float}
\usepackage{soul}
\usepackage{titlesec}
\usepackage{booktabs}
\usepackage{subcaption}

\graphicspath{{Figures/}}

\newcommand\Nr{\ensuremath{ N_{\rm r} }}
\newcommand\Nrf{\ensuremath{ N_{\rm RF} }}
\newcommand\Nw{\ensuremath{ N_{\rm w} }}
\newcommand\Nu{\ensuremath{ N_{\rm u} }}

\newcommand\rmj{\ensuremath{ \mathrm{j} }}
\newcommand\Kcl{\ensuremath{\mathcal{K}}}
\newcommand\ssb{\ensuremath{ \mathbf{s} }}
\newcommand\Cs{\ensuremath{{\mathbb{C}}}}

\newcommand\rmc{\ensuremath{ \mathrm{c} }}
\newcommand\rmd{\ensuremath{ \mathrm{d} }}
\newcommand\rma{\ensuremath{ \mathrm{a} }}
\newcommand\rme{\ensuremath{ \mathrm{e} }}
\newcommand\rms{\ensuremath{ \mathrm{s} }}
\newcommand\rmf{\ensuremath{ \mathrm{F} }}

\newcommand\rmvec{\ensuremath{ \mathrm{Vec} }}
\newcommand\rmtr{\ensuremath{ \mathrm{Tr} }}
\newcommand\rmt{\ensuremath{ \mathrm{t} }}
\newcommand\ew{{\mathbf f}}
\newcommand\mathbfeW{{\mathbf W_\rme \mathbf W_\rma \mathbf W_\rmd }}
\newcommand\eW{{\mathbf F}}
\newcommand\ep{{\mathbf z}}

\def\BibTeX{{\rm B\kern-.05em{\sc i\kern-.025em b}\kern-.08em
T\kern-.1667em\lower.7ex\hbox{E}\kern-.125emX}}

\newtheorem{lemma}{Lemma}

\newtheorem{proposition}{Proposition}

\SetKwInput{Input}{input}
\SetKwInput{Output}{output}

\usepackage{array}
\newcolumntype{L}[1]{>{\raggedright\let\newline\\\arraybackslash\hspace{0pt}}m{#1}}
\newcolumntype{C}[1]{>{\centering\let\newline\\\arraybackslash\hspace{0pt}}m{#1}}
\newcolumntype{R}[1]{>{\raggedleft\let\newline\\\arraybackslash\hspace{0pt}}m{#1}}

\newtheorem{remark}{Remark}

\IEEEoverridecommandlockouts
\def\specialpapernotice#1{\if@confmode%
	\def\@specialpapernotice{{\sublargesize\textit{#1}\vspace*{1em}}}%
	\else%
	\def\@specialpapernotice{{\\*[1.5ex]\sublargesize\textit{#1}}\vspace*{-2ex}}%
	\fi}

\makeatletter\patchcmd{\@makecaption}{\scshape}{}{}{}

\usepackage[all=normal,paragraphs=tight,floats=normal,mathspacing=normal,wordspacing=tight,charwidths=tight,mathdisplays=normal,leading=normal]{savetrees}

\newcommand{\T}{{\scriptscriptstyle\mathsf{T}}}
\renewcommand{\H}{{\scriptscriptstyle\mathsf{H}}}

\allowdisplaybreaks
\begin{document}

\titlespacing{\section}{-0.64 cm}{4pt}{2pt}
\titlespacing{\subsection}{0 cm}{4pt}{2pt}

\title{Tri-Hybrid Beamforming Design for Large-Scale \\MIMO ISAC Systems}
\author{Tianyu Fang, \IEEEmembership{Student Member, IEEE}, Mengyuan Ma, \IEEEmembership{Student Member, IEEE}, Markku Juntti, \IEEEmembership{Fellow, IEEE},\\Inkyu Lee, \IEEEmembership{Fellow, IEEE}, Joonhyuk Kang, \IEEEmembership{Member, IEEE}, and Nhan Thanh Nguyen, \IEEEmembership{Member, IEEE}
 \vspace{-0.8cm}
\thanks{
A part of this work was accepted to be presented at the IEEE Int. Conf. Acoustics, Speech, and Signal Process. (ICASSP), 2026 \cite{fang2026tri}.
\par T. Fang, M. Ma, M. Juntti, and N. T. Nguyen are with the Centre for Wireless Communications, University of Oulu, 90014 Oulu, Finland (e-mail:
\{tianyu.fang, mengyuan.ma, markku.juntti, nhan.nguyen\}@oulu.fi). 
Inkyu Lee is with the School of Electrical Engineering, Korea University, Seoul 02841, South Korea (e-mail: inkyu@korea.ac.kr). 
Joonhyuk Kang are with the Department of Electrical Engineering, Korea Advanced Institute of Science and Technology, Daejeon 34141, South Korea (e-mail: jhkang@ee.kaist.ac.kr).}
}


\maketitle

\begin{abstract}
Tri-hybrid multiple-input multiple-output (MIMO) architectures have been proposed as a promising solution for enabling energy-efficient communications systems in large-scale antenna arrays by replacing conventional antenna arrays in hybrid beamforming (HBF) systems with low-cost dynamic metasurface antennas (DMAs). In this work, we investigate beamforming design for integrated sensing and communications (ISAC) systems based on tri-HBF architectures, with the objective of jointly enhancing the communications sum rate and sensing sum mutual information. Specifically, we formulate a weighted multi-objective optimization problem that balances communications throughput and sensing mutual information, subject to a total transmit power consumption constraint and the physical limitations inherent to tri-HBF architectures. By exploiting the problem structure, we develop an efficient iterative algorithm with closed-form updates to solve the non-convex optimization problem with low complexity. Numerical results are provided to evaluate the performance of the proposed tri-HBF architecture in various setups. The results demonstrate that the proposed tri-HBF architecture can achieve significant improvement in energy efficiency (EE) compared to the conventional hybrid MIMO and DMA-only configurations, at the cost of minor degradations in sum rate and sensing mutual information.
\end{abstract}


\section{Introduction}\label{Sec:intro}

\par Sixth-generation (6G) wireless networks are expected to integrate sensing as a native functionality, enabling simultaneous high-rate communications and high-resolution environmental perception. Such convergence is crucial for applications including autonomous driving, extended reality, and digital twins~\cite{dong2025communication}. Integrated sensing and communications (ISAC) has emerged as a promising paradigm, leveraging shared spectrum, hardware, and signal processing for joint data transmission and environmental sensing~\cite{liu2022integrated}.

The realization of ISAC in 6G is expected to rely on large-scale multiple-input multiple-output (MIMO) technologies, where large antenna arrays provide both high spatial multiplexing gains for communications and fine spatial resolution for sensing~\cite{Lu2024}. However, the deployment of extremely large antenna arrays introduces significant challenges in terms of hardware complexity, power consumption, and implementation cost~\cite{jiang2024THz}. Conventional fully digital and hybrid analog-digital beamforming often face tradeoffs among hardware cost, energy efficiency (EE), and beamforming flexibility in large-scale MIMO systems. To overcome these limitations, the tri-hybrid beamforming (tri-HBF) architecture~\cite{castellanos2025embracing} has recently been proposed. It performs signal processing across the digital, analog, and electromagnetic (EM) domains, where the latter is typically realized by metasurface or reconfigurable antenna technologies~\cite{castellnos2023energy,liu2025reconfigurable,liu2025tri,zheng2025tri,chen2025integrated,li2025tri,cheng2025performance,zhao2025tri}. By reducing the number of active RF chains while maintaining high beamforming gain, tri-HBF offers a cost-effective and energy-efficient solution for large-scale MIMO ISAC systems.

\subsection{Related Works}

The design of ISAC systems has attracted significant attention in recent years, with a variety of beamforming architectures proposed to balance performance, cost, and complexity. The earliest and most fundamental benchmark is fully digital beamforming, which provides maximum flexibility in signal shaping and multi-user support, thereby enabling high-precision target sensing and interference mitigation~\cite{liu2018mu,liu2020joint,fang2025optimal,zhang2026rotatable}. Beamforming strategies for joint MIMO radar–communications systems have been widely studied under both separated and shared antenna deployments. Shared waveform designs based on weighted optimization can significantly outperform separated schemes while reducing system complexity~\cite{liu2018mu,liu2020joint}. However, such semidefinite relaxation (SDR)-based formulations often incur prohibitive high computational complexity, especially for large-scale arrays. To enable scalable ISAC transmission, massive MIMO systems have motivated low-complexity designs that combine linear precoders with predesigned beamformers such as maximum-ratio transmission (MRT) or zero-forcing (ZF), together with tractable power allocation schemes~\cite{liao2024power}. In addition, efficient optimization frameworks, including alternating direction method of multipliers (ADMM)-based methods~\cite{tang2025dual,Guo2023FP-ADMM}, fractional programming (FP) approaches~\cite{zhu2023integrated,chen2025fast}, and successive convex approximation (SCA) techniques~\cite{fang2025low,fang2024beamforming}, have been developed to address various ISAC beamforming problems.

While fully-digital designs provide useful performance benchmarks, their scalability to extremely large antenna arrays is limited by hardware cost and energy consumption, which has motivated research on HBF. Compared to fully digital beamformers, HBF architectures reduce hardware cost by employing fewer RF chains, while still enabling high beamforming gains \cite{wang2022partially,qi2022hybrid,chen2025mimo-dfrc,nhan2023deep,nhan2023multiuser}. For instance,\cite{wang2022partially} jointly optimized the analog and digital precoders to minimize the Cramér–Rao lower bound (CRLB) under user signal-to-interference-plus-noise ratio (SINR) constraints. Similarly,\cite{chen2025mimo-dfrc} investigated a fully connected HBF structure, where transmit patterns were designed to maximize the ratio of mainlobe-to-sidelobe power under both SINR and power budget constraints. Beyond optimization-based designs, learning-driven methods have also been proposed. For example,~\cite{nhan2023deep} introduced a deep unfolding approach for HBF in massive MIMO-ISAC systems, achieving improved communications–sensing performance tradeoffs, faster convergence, and significantly reduced computational complexity compared with conventional optimization approaches.

Beyond conventional antenna arrays, recent works have begun exploring dynamic metasurface antenna (DMA) architectures, which extended HBF by incorporating EM domain processing. Unlike arrays with half-wavelength spacing, DMAs can achieve denser element packing and richer spatial manipulation, making them promising for ISAC \cite{gavras2023full,Bayraktar2024near,zhu2025the}. For example,~\cite{Bayraktar2024near} studied the design of precoders and combiners for near-field full-duplex ISAC systems, exploiting the reconfigurability of DMAs at the base station (BS). A related line of work integrated DMAs with RISs, where joint optimization frameworks were developed to maximize symbiotic transmission rates while satisfying sensing and hardware constraints~\cite{zhu2025the}. While these studies demonstrated the potential of DMA-based ISAC, research in this direction remains at an early stage. Existing works primarily focus on specific system settings or simplified architectures, while the joint design of scalable beamforming algorithms, hardware-constrained DMA control, and integrated sensing–communication performance tradeoffs is still largely open.

\subsection{Contributions}
 
In this work, we investigate a tri-HBF architecture for monostatic large-scale MIMO ISAC system. The proposed architecture integrates digital beamformer, analog phase shifters, and EM domain processing enabled by DMAs, thereby allowing denser radiating element packing and less power consumption. We consider a downlink scenario where the BS simultaneously serves multiple communications users while performing multi-target sensing. We aim to jointly optimize the tri-HBF transceiver to balance communications sum rate and sum sensing mutual information (MI), subject to practical hardware constraints on transmit power, DMA responses, and analog phase shifters structures.

The main contributions of this paper are summarized as follows:
\begin{itemize}
    \item We consider an large-scale MIMO ISAC system that employs a tri-HBF architecture, which combines DMA-enabled EM domain processing with conventional analog and digital beamforming. While DMA, hybrid beamforming, and ISAC have each been studied individually, their joint integration in a unified ISAC transmission framework remains largely unexplored. Due to the inherently conflicting performance requirements of communications and sensing, the beamforming design is formulated as a multi-objective optimization problem. More specifically, we maximize the weighted sum of the communications and sensing objectives subject to transmit power constraints as well as practical DMA and analog hardware limitations.
    
    \item The resulting optimization problem is highly non-convex and involves multiple coupled variables. To efficiently solve it, we first eliminate the explicit transmit power constraint by showing that the optimal solution always utilizes the full available transmit power. To further handle the non-convex objective function, we develop a FP-based framework that transforms the original objective into a more tractable surrogate form while guaranteeing equivalence. Within this framework, we propose a block coordinate descent (BCD) algorithm that alternately optimizes the digital beamformer, analog beamformer, and DMA coefficients. Notably, each subproblem admits a closed-form update, which avoids the need for generic numerical solvers, thereby significantly reducing computational complexity. 
    
    \item Extensive simulation results are presented to evaluate the proposed tri-HBF large-scale MIMO ISAC system and to compare it with existing architectures, including fully digital, conventional hybrid, and DMA-based designs. The results demonstrate that the proposed tri-HBF architecture achieves superior EE in both communications- and sensing-dominated operating regions, at the cost of minor degradations in sum rate and sensing mutual information due to the inherent physical constraints of the DMA.
\end{itemize}

Unlike~\cite{chen2025integrated}, which considered tri-hybrid beamforming for ISAC with electromagnetically
reconfigurable antenna (ERA), this work adopts DMA as the EM processing layer, leading to fundamentally different structural constraints on EM-domain modeling. In addition, we impose an antenna output power constraint rather than the antenna input power constraint used in~\cite{chen2025integrated}, which accurately capturing the actual radiated power and the signal transformation induced by the reconfigurable EM structure~\cite{castellanos2025embracing}. Furthermore, while~\cite{chen2025integrated} focuses on a single-target scenario and optimizes SCNR, we consider a multi-target setting and adopt the sum MI as an information-theoretic sensing metric. To the best of our knowledge, this is the first work to develop a systematic and scalable optimization framework for DMA-based tri-HBF in large-scale MIMO ISAC systems under practical hardware constraints and dual-functional design objective.

\textit{Organization:}  The rest of this paper is organized as follows. Section \ref{System model} introduces the system model and problem formulation. Section \ref{sec: algorithm} presents the proposed method along with a thorough analysis of its optimality and convergence. Simulation results and numerical discussions are provided in Section \ref{Sec:simulation} to demonstrate the effectiveness of the proposed algorithms. Finally, Section \ref{Sec:concu} concludes the paper and discusses potential directions for future works.

\textit{Notation:} Vectors and matrices are represented by lowercase and uppercase boldface letters, respectively. The sets of complex- and real-valued numbers are respectively denoted by $\mathbb{C}$ and $\mathbb{R}$, while $\mathbb{R}_{++}$ represents the set of positive real numbers. We use $\mathbb{E}[ \cdot ]$ to represent the expectation of a random variable. The magnitude of a complex number is written as $|\cdot|$. The circularly symmetric complex Gaussian  distribution with zero mean and variance $\sigma^2$ is represented as $\mathcal{CN}(0,\sigma^2)$. The conjugate, transpose and conjugate transpose operators are denoted by $(\cdot)^*$, $(\cdot)^\T$ and $(\cdot)^\H$, respectively, while $\mathrm{diag}{\{\mathbf{y}\}}$ denotes a diagonal matrix whose main diagonal entries are elements of vector $\mathbf{y}$. The element-wise (Hadamard) product is denoted by $\circ$, and the Kronecker product is denoted by $\otimes$. Finally, $[\mathbf{X}]_{i,j}$ denotes the entry in the $i$-th row and $j$-th column.

\section{System Model and Problem Formulation}\label{System model}

\begin{figure}[t!]
\center
\includegraphics[width=0.5\textwidth]{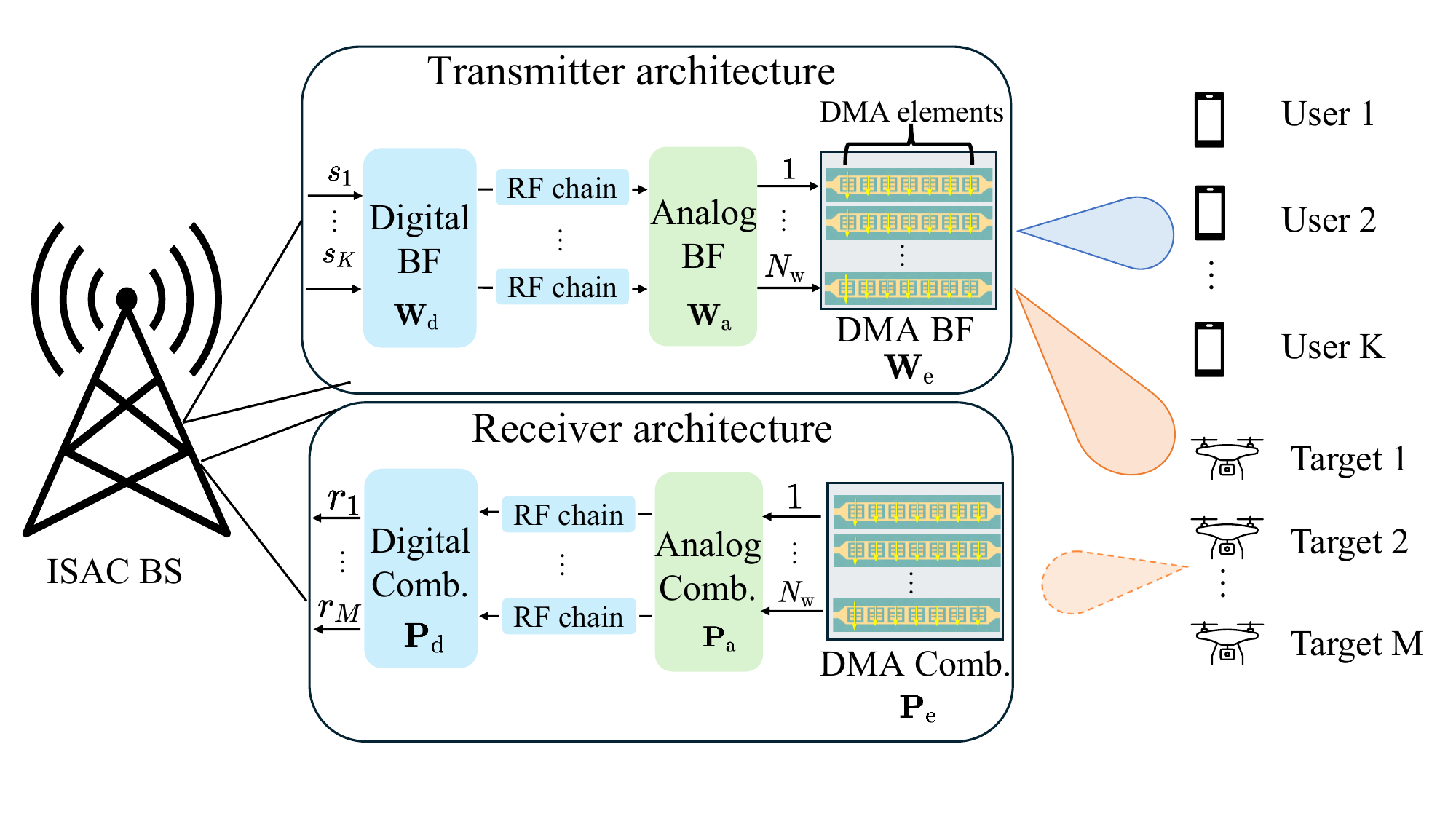}
    \vspace{-5mm}
    \caption{Illustration of the proposed tri-HBF ISAC system. The base station employs digital (Digital BF), analog (Analog BF), and DMA-based electromagnetic beamforming layers for transmission, along with corresponding combining modules for reception.}
    \label{fig:system model}
\end{figure}
We consider a downlink monostatic millimeter-wave (mmWave) ISAC system\footnote{While the formulation focuses on a narrowband setting, it can be extended to wideband orthogonal frequency division multiplexing (OFDM) systems by keeping the DMA and analog beamformers frequency-flat and optimizing the digital beamformers on a per-subcarrier basis. The resulting problem decomposes across subcarriers, and the proposed algorithmic framework can be applied with minor modifications.}, where a base station (BS) equipped with a dynamic metasurface antenna (DMA)-based tri-hybrid beamforming architecture, as illustrated in Fig.~\ref{fig:system model}, simultaneously serves $K$ single-antenna communication users and performs sensing. Let $\Kcl \triangleq \{1,\ldots,K\}$ denote the set of users. Meanwhile, the transmitted signals are leveraged to probe $M$ point-like sensing targets, indexed by $\mathcal{M} \triangleq \{1,\ldots,M\}$, in the presence of $C$ clutter scatterers.


\subsection{Signal Model}

Let $ \ssb=[s_1,\ldots, s_{K}]^\T \in \Cs^{K\times 1} $ denote the vector of symbols intended for the $K$ users such that $\mathbb E[\mathbf s\mathbf s^\H]=\mathbf I_K$. The symbols are first processed by the digital beamformer $\mathbf W_\rmd=[\mathbf w_{\rmd 1},\ldots,\mathbf w_{\rmd K}]\in\mathbb C^{\Nrf\times K}$, where $ \Nrf$ denotes the number of RF chains. The signal is further processed by the analog beamformer $\mathbf W_\rma\in\mathbb C^{\Nw\times \Nrf} $, where $\Nw$ refers to the number of waveguides of the DMA. To fully exploit the potential of the tri-HBF transmitter, we assume a fully connected analog beamformer implemented by phase shifters~\cite{yu2016alternating}. This imposes the constraint:
\begin{align}
    \mathcal C_{\rma}\triangleq\{\mathbf W_\rma\big| \left|[\mathbf W_\rma]_{i,j}\right|=1, \forall i\in\mathcal N_{\rm w}, \forall j\in\mathcal N_{\rm rf}\},
\end{align}
where $\mathcal N_{\rm w}\triangleq\{1,\ldots,\Nw\}$ and $\mathcal N_{\rm rf}\triangleq\{1,\ldots,\Nrf\}$.

The DMA is realized as a planar array comprising $\Nw$ waveguides, each containing $\Nu$ reconfigurable radiating elements. There are $\Nr=\Nw\Nu$ radiating elements in total, indexed by $\mathcal N_{\rm r}=\{1,\ldots,\Nr\}$. Each radiating element captures the signal from its corresponding waveguide, applies a tunable phase shift, and transmits it via the wireless channel. Due to the single input port of each waveguide, the array naturally partitions into subarrays, resulting in a block-diagonal structure in the overall DMA beamformer:
\begin{equation}
    \mathbf W_\rme=\mathrm{blkdiag}\left(\mathbf w_{\rme 1},\ldots,\mathbf w_{\rme \Nw}\right),
\end{equation}
 where the $m$-th element of $\mathbf w_{\rme n}$ is given by $   w_{\rme m, n}\triangleq q_{m,n}\kappa_{m,n}, \forall m\in\mathcal N_{\rm u}, \forall n\in\mathcal N_{\rm w}$ with $\mathcal N_{\rm u}\triangleq\{1,\ldots,\Nu\}$. Here, $q_{m,n}\in\mathbb C$ represents the constant propagation gain from the input port of $n$-th waveguide to the $m$-th radiator. It is modeled as
\begin{align}
    q_{m,n}=e^{-l_{m,n}(\vartheta+\rmj \varpi)},
\end{align}
 where $l_{m,n}$ denotes the physical distance between the input port and the corresponding radiating element, $\vartheta$ is the attenuation coefficient, $\rmj$ denotes imaginary unit, and $\varpi$ is the wavenumber. Furthermore, $\kappa_{m,n}\in\mathbb C$ is the tunable coefficient applied by the $m$-th radiating element, which satisfies the Lorentzian constraint \cite{nir2019dynamic}:
\begin{align}
    \kappa_{m,n}=\frac{{\rmj}+e^{\rmj\psi_{m,n}}}{2}, \psi_{m,n}\in[0,2\pi).
\end{align}
For notational clarity, we introduce a compact representation of the DMA coefficient matrix $\mathbf W_\rme$ as
\begin{align}\label{M compact}
    \mathbf W_\rme=\frac{1}{2}(\rmj\mathbf Q+\mathbf \Psi\circ\mathbf Q),
\end{align}
where $\circ$ represents the element-wise product, $\mathbf Q\in\mathbb C^{\Nr\times\Nw}$ collects all propagation gain coefficients, defined as
\begin{align}
    [\mathbf Q]_{m,n}=
    \begin{cases}
        q_{m,n},&m\in\mathcal N_{\rm r}, n=\lfloor \frac{m-1}{\Nu}\rfloor+1,\\
        0,  & \text{otherwise}.
    \end{cases}
\end{align}
Similarly,  $\mathbf\Psi\in\mathbb C^{\Nr\times\Nw}$ collects the phase shifter coefficients of the radiating elements, given by
\begin{align}\label{Psi compact}
    [\mathbf \Psi]_{m,n}=
    \begin{cases}
        e^{{\rmj}\psi_{m,n}},&m\in\mathcal N_{\rm r}, n=\lfloor \frac{m-1}{\Nu}\rfloor+1,\\
        0,  & \text{otherwise}.
    \end{cases}
\end{align}
Accordingly, the feasible set of the DMA beamformer coefficients is defined as
\begin{equation}
\mathcal C_{\rme} \triangleq \{\mathbf W_\rme \mid \mathbf W_\rme 
\text{ satisfies the DMA constraints}\}.
\end{equation}

Finally, the overall transmitted signal is given by
\begin{equation}
    \mathbf x=\sum_{k=1}^K \ew_k s_k=\eW\mathbf s,
\end{equation}
where $\ew_k\triangleq \mathbf W_\rme\mathbf W_\rma\mathbf w_{\rmd k}$ denotes the $k$-th effective beamforming vector and $  \eW\triangleq[\ew_1,\ldots,\ew_K]=\mathbf W_\rme\mathbf W_\rma\mathbf W_\rmd$. The total transmit power is given by $\mathbb E\left[\mathbf x^\H\mathbf x \right]=\|\eW\|^2_\rmf $, where $\|\cdot\|_\rmf$ denotes the Frobenius norm.

\subsubsection{Communications Model}\label{Sec:communication model}
For communications, the Saleh–Valenzuela scattering model is adopted to characterize mmWave channels~\cite{ayach2014spatially}. Specifically, the channel vector for user $k$ is expressed as
\begin{align}
\mathbf h_{k} = \sqrt{\frac{\Nr}{L}} \sum_{l=1}^{L} \alpha_{l,k}\mathbf a_{\rmt}(\theta_{l,k}, \phi_{l,k}),
\end{align}
where $L$ denotes the number of propagation paths, $\alpha_{l,k}$ is the complex gain of the $l$-th path for user $k$, and $\mathbf a_\rmt(\theta_{l,k}, \phi_{l,k})$ represents the transmit array steering vector corresponding to the azimuth angle $\theta_{l,k} \in [-\pi, \pi]$ and elevation angle $\phi_{l,k} \in [-\pi/2, \pi/2]$.

The received signal at communications user $k $ is given by
\begin{equation}
y_{\mathrm{c}k} = \mathbf h_k^\H \ew_k s_k + \sum_{i=1, i \neq k}^{K} \mathbf h_k^\H \ew_i s_i + n_k,
\end{equation}
where $n_k \sim \mathcal{CN}(0, \sigma_k^2)$ denotes the additive white Gaussian noise (AWGN) with variance $\sigma_k^2$. In this work, to focus on demonstrating the performance of tri-HBF, we assume that the BS has perfect communications channel state information (CSI). Accordingly, the achievable SINR for user $k \in \mathcal{K}$ is given by
\begin{equation} \label{rateo}
\gamma_{\rmc k} =\frac{|\mathbf h_k^\H \ew_k|^2}{\sum_{i=1, i \neq k}^{K} |\mathbf h_k^\H \ew_i|^2 + \sigma_k^2}.
\end{equation}
In the subsequent design and optimization, we employ the sum rate $\sum_{k=1}^K\log(1+\gamma_{\rmc k})$ as the communications performance metric.

\subsubsection{Radar Model}\label{sec:sensing model}
For radar sensing, the received echo signal at the BS can be expressed as
\begin{equation} \label{sensing signal matrix}
\mathbf y_{\mathrm{s}} =\sum_{m=1}^M \mathbf G_m \mathbf x+\sum_{j=M+1}^{M+C}\mathbf G_j \mathbf x+\mathbf n_\rms,
\end{equation}
where $\mathbf G_i \triangleq \alpha_{\rms i}\mathbf a_{\rm r}(\theta_i, \phi_i)\mathbf a_\rmt^\H(\theta_i, \phi_i), i=\{1,\ldots, M+C\}$ represents the line-of-sight (LoS) two-way channel matrix corresponding to targets and clutters, and $\mathbf n_\rms\sim \mathcal{CN}(0, \sigma_\rms^2\mathbf I)$. Here, $\alpha_{\rms i}$ denotes the reflection coefficient, which accounts for the radar cross-section and the round-trip path loss, while $\mathbf a_{\rm r}(\cdot)$ and $\mathbf a_\rmt(\cdot)$ represent the normalized receive and transmit array steering vectors, respectively.

Assuming that the radar receiver adopts the same tri-HBF architecture, the received echo signal is first captured by the DMA, characterized by $\mathbf P_\rme\in\mathcal C_{\rme}$ and then processed by the analog combiner $\mathbf P_\rma\in\mathcal C_{\rma}$. Let $\mathbf p_{\rmd m}\in\mathbb C^{\Nrf \times 1}$ denote the digital combining vector associated with target $m$, the corresponding effective combiner is defined as
\begin{equation}
    \ep_m\triangleq \mathbf P_\rme\mathbf P_\rma\mathbf p_{\rmd m}.
\end{equation}
The resulting combined signal is given by
\begin{align}
r_m &= \ep_m^\H\mathbf y_\rms\\
&= \ep_m^\H\mathbf G_m \eW \mathbf s+\ep_m^\H\left(\sum_{j\neq m}^{M+C}\mathbf G_j \eW \mathbf s+\mathbf n_\rms\right).
\end{align}
Accordingly, the signal-to-clutter-plus-noise ratio (SCNR) for target $m$ is defined as
\begin{align}\label{SCNR}
    \gamma_{\rms m}=\frac{\|\ep_m^\H\mathbf G_m\eW \|_\rmf^2}{\sum_{j=1,j\neq m}^{M+C}\|\ep_m^\H\mathbf G_j\eW\|_\rmf^2 +\sigma_\rms^2\|\ep_m^\H\|^2_\rmf }.
\end{align}
The sensing performance is then quantified by the sum mutual information~\cite{peng2024mutual}, defined as $\sum_{m=1}^M\log(1+\gamma_{\rms m})$.

\subsection{Problem Formulation}
We aim to jointly design the tri-HBF transceiver beamformers and combiners, consisting of the transmit beamformers $\{\mathbf W_{\rme},\mathbf W_{\rma},\mathbf W_{\rmd}\}$ and the receive combiners $\{\mathbf P_{\rme},\mathbf P_{\rma},\mathbf P_{\rmd}\}$, to simultaneously achieve high communications sum rate and sensing total mutual information. To this end, we adopt a widely used scalarization approach, wherein a weighted sum of the communications and sensing objectives is optimized under varying the weight factors. The resulting optimization problem is formulated as
\begin{subequations}\label{P1}
    \begin{align}
    \max_{ \bm\Omega}\,\, &\delta_\rmc \sum_{k=1}^K \log(1+\gamma_{\rmc k})+\delta_\rms \sum_{m=1}^M \log(1+\gamma_{\rms m})\\
    \label{P1C1}\text{s.t.}\,\, & \mathbf W_{\rme},\mathbf P_{\rme}\in\mathcal C_{\rme},\\
   \label{P1C3} &\mathbf W_{\rma},\mathbf P_{\rma}\in\mathcal C_{\rma},\\    \label{P1C4}&\|\eW\|_\rmf^2\leq P_\rmt,
\end{align}
\end{subequations}
where $ \delta_\rmc \geq 0$ and $\delta_\rms\geq 0$ are weighting factors that balance the communications and sensing performance, $P_\rmt$ denotes the total transmit power budget, and $\bm\Omega\triangleq\{ \mathbf W_\rmd,\mathbf W_{\rma},\mathbf W_{\rme},\mathbf P_\rmd,\mathbf P_{\rma},\mathbf P_{\rme}\}$ collects all optimization variables. In this problem, constraint \eqref{P1C1} captures the physical constraints associated with the DMA radiating elements, \eqref{P1C3} represents the phase shifter hardware constraints, and \eqref{P1C4} enforces the transmit power limitation. Problem \eqref{P1} is inherently NP-hard due to the non-convex fractional forms of the SINR and SCNR in \eqref{rateo} and \eqref{SCNR}, the unit-modulus constraints, and the intricate coupling among the optimization variables. To tackle these challenges, we employ the FP framework to develop an efficient iterative algorithm that guarantees convergence to a stationary point of \eqref{P1}, as elaborated in the next section.

\section{Proposed Tri-HBF Design Framework}
\label{sec: algorithm}

To efficiently solve~\eqref{P1}, we first eliminate the transmit power constraint by identifying and leveraging a key property, which reveals that the optimal solution must fully utilize the available transmit power budget. Then, we introduce a FP framework~\cite{Shen2018fractional} that yields a tractable surrogate objective. Finally, we adopt the BCD method~\cite{xu2013block} to alternatively optimize the tri-HBF's digital, analog, and DMA coefficients.

\subsection{Problem Reformulation}
\label{Sec:Reformulation}

We found that the coupling beamfomring matrices in the power constraint \eqref{P1C3}  complicates the development of efficient optimization algorithms. To address this issue, we first introduce the following proposition.

\begin{proposition}\label{pro:full power}
    Any locally optimal solution of problem \eqref{P1} must result in full transmit power consumption, i.e., $\|\eW\|_\rmf^2 =  P_{\mathrm{t}}$. Thus, problem \eqref{P1} is equivalent to
   \begin{subequations}\label{P1full}
    \begin{align}
    \max_{ \bm\Omega}\,\, &\delta_\rmc\sum_{k=1}^K \log(1+\gamma_{\rmc k})+\delta_\rms \sum_{m=1}^M \log(1+\gamma_{\rms m})\\
    \label{P1fullC1}\text{s.t.}\,\, &\|\eW\|_\rmf^2= P_\rmt,\\
    & \eqref{P1C1},\eqref{P1C3}.
    \end{align}
\end{subequations}
\end{proposition}
 \begin{IEEEproof}
   See Appendix \ref{property:full power}.
\end{IEEEproof}
\smallskip
This property has been established in our prior work~\cite{fang2025optimal} for the scenario with sensing metric measured by CRLB and fully digital beamforming. We herein extend it to tri-HBF systems. Proposition~\ref{pro:full power} indicates that the transmit beamforming matrices can be constrained to use the full power budget without any loss of optimality. This significantly reduces the solution space and simplifies the optimization problem. By using Proposition~\ref{pro:full power} and exploiting the fractional structure of the SINR and SCNR in \eqref{rateo} and \eqref{SCNR} respectively, we obtain the following proposition.
\begin{proposition}\label{pro:equivalence}
The original problem \eqref{P1} is equivalent to
    \begin{subequations}\label{P2}
    \begin{align}
    \label{P2obj}\max_{ \bm\Omega}\,\, &\delta_\rmc\sum_{k=1}^K \log(1+\widehat{\gamma}_{\rmc k})+\delta_\rms \sum_{m=1}^M \log(1+\widehat{\gamma}_{\rms m})\\
    \text{s.t.}\,\, 
    & \eqref{P1C1},\eqref{P1C3},
\end{align}
\end{subequations}
where 
\begin{align}
   & \widehat \gamma_{\rmc k} =  \frac{|\mathbf h_k^\H \ew_k|^2}{\sum_{i=1, i \neq k}^{K} |\mathbf h_k^\H \ew_i|^2 + \bar{\sigma}_k^2\|\eW\|_\rmf^2},\\
   & \widehat{\gamma}_{\rms m}=\frac{\|\ep_m^\H\mathbf G_m\eW \|_\rmf^2}{\sum_{j=1,j\neq m}^{M+C}\|\ep_m^\H\mathbf G_j\eW \|_\rmf^2 +\bar{\sigma}_\rms^2\|\ep_m^\H\|^2_\rmf\|\eW\|_\rmf^2 }.
\end{align}
with $\bar{\sigma}_k\triangleq \sigma_k/\sqrt{P_\rmt} $ and $\bar{\sigma}_s\triangleq \sigma_k/\sqrt{P_\rmt} $.

\end{proposition}

\begin{IEEEproof}
   See Appendix \ref{App:equivalence}.
\end{IEEEproof}
\smallskip
Proposition \ref{pro:equivalence} implies that problem \eqref{P1} can be effectively addressed by solving its reformulated version \eqref{P2}. Once a solution $ \{\mathbf W_\rme^\diamond, \mathbf W_\rma^\diamond,\mathbf W_\rmd^\diamond\}$ to problem \eqref{P2} is obtained, a corresponding feasible solution $ \mathbf W_\rmd^\ddagger$ to the original problem \eqref{P1} can be constructed as
\begin{equation}\label{scale W}
        \mathbf W_\rmd^\ddagger=\sqrt{\frac{P_\rmt}{\|\mathbf W_{\rme}^\diamond\mathbf W_{\rma}^\diamond\mathbf W_\rmd^\diamond\|_\rmf^2}}\mathbf W_\rmd^\diamond.
\end{equation}
Accordingly, the two formulations are equivalent in the sense that solving \eqref{P2} and applying the above transformation yields a valid solution to the original problem \eqref{P1}.
\begin{remark}
The reformulation from problem \eqref{P1} to \eqref{P2} brings the following advantages:
\begin{itemize}
    \item The optimization with respect to the digital beamforming matrix $\mathbf{W}_\rmd$ becomes unconstrained, which leads to closed-form solutions and significantly reduces computational complexity.
    \item The original coupling between the analog beamforming matrix, the DMA beamforming matrix, and the digital beamforming matrix in the power constraint is removed. This decoupling greatly simplifies the optimization of the analog and DMA components, allowing for more tractable algorithm design, as will be presented next.
\end{itemize}
\end{remark}

\subsection{Fractional Programming Method} \label{sec:FP method}

We aim to further transform \eqref{P2} into a more tractable reformulation. To this end, we present the following two lemmas.

\begin{lemma}[\!\!\cite{Shen2018fractional}]\label{lemma:LDT}
    Given ratios $\frac{c_k(x)}{d_k(x)}$ with $c_k(x)\geq 0$ and $d_k(x)>0$ for $k=1,\ldots,K$, the sum-of-logarithmic-ratios maximization problem:
    \begin{equation}
        \max_x\,\, \sum_{k=1}^K \log\left(1+\frac{c_k(x)}{d_k(x)}\right)
    \end{equation}
    is equivalent to
    \begin{equation}
        \max_{x,\{\eta_i\}_{i=1}^K}\,\, \sum_{k=1}^K\log(1+\eta_k)-\eta_k+\frac{(1+\eta_k)c_k(x)}{c_k(x)+d_k(x)}.
    \end{equation}
\end{lemma}

\begin{lemma}[\!\!\cite{Shen2018fractional}] \label{lemma:QT}
    Given ratios $\frac{|a_k(x)|^2}{b_k(x)}$ with $a_k(x)\in\mathbb C$ and $b_k(x)>0$ for $k=1,\ldots,K$, the sum-of-ratios maximization problem
    \begin{equation}
        \max_x \,\,\sum_{k=1}^K\frac{|a_k(x)|^2}{b_k(x)}
    \end{equation}
    is equivalent to
    \begin{equation}
        \max_{x,\{\beta_i\}_{i=1}^K}\,\,\sum_{k=1}^K2\Re\{\beta_k^*a_k(x)\}-|\beta_k|^2b_k(x).
    \end{equation}
\end{lemma}

Applying Lemma \ref{lemma:LDT} to the term $\log(1+\widehat{\gamma}_{\rmc k})$ in \eqref{P2obj} and setting $c_k(x)=|\mathbf h^\H_k\ew_k|^2$ and $d_k(x)=\sum_{i=1, i \neq k}^{K} |\mathbf h_k^\H \ew_i|^2 + \bar{\sigma}_k^2\|\eW\|_\rmf^2$, we construct a surrogate function of $\log(1+\widehat{\gamma}_{\rmc k})$ as follows:
\begin{align}\label{surrogate1}
\tilde{R}_{\rmc k}\triangleq\log(1&+\eta_{\rmc k}) -\eta_{\rmc k}+\frac{(1+\eta_{\rmc k})|\mathbf h^\H_k\ew_k|^2}{\sum_{i=1}^{K} |\mathbf h_k^\H \ew_i|^2+\bar{\sigma}_k^2\|\eW\|_\rmf^2},
\end{align}
where $\eta_{\rmc k}$ is an introduced auxiliary variable. Then, the ratio in $\tilde{R}_{\rmc k}$ can be further decoupled using Lemma \ref{lemma:QT} with $a_k(x)=\mathbf h_k^\H\ew_k $ and $b_k(x)=\sum_{i=1}^{K} |\mathbf h_k^\H \ew_i|^2 + \bar{\sigma}_k^2\|\eW\|_\rmf^2$, leading to the following surrogate function of $\tilde{R}_{\rmc k}$:
\begin{align}\label{surrogate2}
\bar{R}_{\rmc k}&\triangleq\log(1+\eta_{\rmc k}) -\eta_{\rmc k}+2\sqrt{1+\eta_{\rmc k}}\Re\{\mathbf h^\H_k\ew_k \beta_{\rmc k}^*\}\nonumber\\
    &-|\beta_{\rmc k}|^2\left(\sum_{i=1}^{K} |\mathbf h_k^\H \ew_i|^2 + \bar{\sigma}_k^2\|\eW\|_\rmf^2\right),
\end{align}
where $\beta_{\rmc k}$ is another auxiliary variable. Similarly, applying Lemma \ref{lemma:LDT} and \ref{lemma:QT} to $\log(1+\widehat{\gamma}_{\rms m})$ in \eqref{P2obj}, we construct the surrogate function
\begin{align}\label{surrogate3}
    \bar{R}_{\rms m}&\triangleq \log(1+\eta_{\rms m})-\eta_{\rms m}+2\sqrt{1+\eta_{\rms m}}\Re\{\ep_m^\H\mathbf G_m\eW\bm\beta_{\rms m}^\H \}\nonumber\\
    &-\|\bm\beta_{\rms m}\|^2_\rmf\left(\sum_{j=1}^{M+C}\!\!\|\ep^\H_m\mathbf G_j\eW\|_\rmf^2+\bar{\sigma}_\rms^2\|\ep_m\|_\rmf^2\|\eW\|_\rmf^2 \!\right)\!.
\end{align}

Using surrogate functions \eqref{surrogate2} and \eqref{surrogate3}, problem \eqref{P2} can be reformulated as
  \begin{subequations}\label{P3}
    \begin{align}
    \label{P3obj}\max_{\bm \Omega, \bm\eta,\bm\beta}\,\, &\delta_\rmc\sum_{k=1}^K \bar R_{\rmc k}+\delta_\rms \sum_{m=1}^M \bar R_{\rms m}\\
    \text{s.t.}\,\, 
    & \eqref{P1C1},\eqref{P1C3},
\end{align}
\end{subequations}
where $\bm\eta\triangleq \{\eta_{\rmc 1},\ldots, \eta_{\rmc K},\eta_{\rms 1},\ldots,\eta_{\rms M}\}$ and $\bm\beta\triangleq\{\beta_{\rmc 1},\ldots,\beta_{\rmc K},\bm\beta_{\rms 1},\ldots,\bm\beta_{\rms M}\}$ collects all the auxiliary variables variables. Although the joint optimization of all variables in problem \eqref{P3} is still intractable, it facilitates a multi-block algorithm that allows for optimizing one block of variables while holding the others fixed. Consequently, we employ the BCD method to optimize the variables iteratively, as presented next.  

\subsection{BCD Framework for Solving \eqref{P3}}
\label{sec: BCD framework}
To facilitate the design, we reformulate the objective function in \eqref{P3obj} into a more compact and explicit representation. Specifically, the communications-related term is expressed as
\begin{align}
    \sum_{k=1}^K \bar{R}_{\rmc k}&=\bar C_\rmc+ 2\Re\{\rmtr(\eW\bm\Sigma_1^\H\mathbf H^\H)\}-S_\rmc\|\eW\|_\rmf^2\nonumber\\
    &-\rmtr(\eW\eW^\H\mathbf H\mathbf \Sigma_2\mathbf H^\H),
\end{align}
where the coefficient scalars and matrices are defined as
\begin{align*}
    &\bar C_{\rmc}\triangleq \sum_{k=1}^K \log(1+\eta_{\rmc k})-\eta_{\rmc k},\,\,
    S_\rmc\triangleq \sum_{k=1}^K \bar{\sigma}_k^2|\beta_k|^2,\\
    &\bm\Sigma_1\triangleq \mathrm{diag}(\sqrt{1+\eta_{\rmc 1}}\beta_{\rmc 1},\ldots,\sqrt{1+\eta_{\rmc K}}\beta_{\rmc K}),\\
    &\bm\Sigma_2\triangleq \mathrm{diag}(|\beta_{\rmc 1}|^2\ldots,|\beta_{\rmc K}|^2),\,\,\mathbf H\triangleq [\mathbf h_1,\ldots,\mathbf h_K].
\end{align*}
Similarly, the sensing-related term can be written as
\begin{align}
    \sum_{m=1}^M \bar{R}_{\rms m}&\!= \!\bar C_{\rms}+2\Re\{\eW\bm\Sigma_3\}\!-\!S_\rms\|\eW\|_\rmf^2-\mathrm{Tr}(\eW\eW^H\bm\Sigma_4),
\end{align}
where 
\begin{align*}
    &\bar C_{\rms}\triangleq \sum_{m=1}^M\! \log(1\!+\!\eta_{\rms m})\!-\!\eta_{\rms m},
    S_\rms\triangleq \sum_{m=1}^M \bar{\sigma}_\rms^2\!\|\bm\beta_{\rms m}\|_\rmf^2\|\ep_m\|_\rmf^2,\\
    &\bm\Sigma_3\triangleq \sum_{m=1}^M \sqrt{1+\eta_{\rms m}}\bm\beta_{\rms m}^\H\ep_m^H\mathbf G_m,\\
    &\bm\Sigma_4\triangleq \sum_{m=1}^M \sum_{j=1}^{M+C} \|\bm\beta_{\rms m}\|^2_\rmf \mathbf G_j\ep_m\ep_m^\H\mathbf G_j^\H.
\end{align*}
 By combining the communications and sensing terms, we reformulate problem \eqref{P3} as
\begin{subequations}\label{P4}
    \begin{align}
    \max_{\bm\Omega,\bm\eta,\bm\beta}\quad &2\Re\{\rmtr(\eW\mathbf C_1^\H)\}-\rmtr(\eW\eW^\H\mathbf C_2)\\
    \text{s.t.}\,\, 
    & \eqref{P1C1},\eqref{P1C3},
\end{align}
\end{subequations}
where
\begin{subequations}\label{eq:C12}
    \begin{align}
    &\mathbf C_1\triangleq \delta_\rmc\mathbf H\bm\Sigma_1+\delta_\rms \bm\Sigma_3^\H, \\
    &\mathbf C_2\triangleq \mathbf \delta_\rmc \mathbf H\bm\Sigma_2\mathbf H^\H+\delta_\rms \bm\Sigma_4+(\delta_\rmc S_\rmc+\delta_\rms  S_\rms)\mathbf I_{\Nr}.
\end{align}
\end{subequations}
In the following, we present the solution to \eqref{P4} using a BCD approach.

\subsubsection{Updating Auxiliary Variables}
Given $\bm\Omega$, the subproblems with respect to $\{\bm\eta,\bm\beta\}$ are unconstrained and convex. By examining their first-order optimality conditions, the optimal solutions for $\{\bm\eta,\bm\beta\}$ can be explicitly obtained in closed-forms as follows:
\begin{align}  \label{update auxilary variables}  \eta_{\rmc k}^\star&=\widehat{\gamma}_{\rmc k},\quad \eta_{\rms m}^\star=\widehat{\gamma}_{\rms m}\\
    \beta_{\rmc k}^\star&=\frac{\sqrt{1+\eta_{\rmc k}}\mathbf h^\H_k\ew_k}{\sum_{i=1}^{K} |\mathbf h_k^\H \ew_i|^2 + \bar{\sigma}_k^2\|\eW\|_\rmf^2},\\
    \bm\beta_{\rms m}^\star&=\frac{\sqrt{1+\eta_{\rms m}}\ep_m^\H\mathbf G_m\eW}{\sum_{j=1}^{M+C}\|\ep_m^\H\mathbf G_j\eW\|_\rmf^2+\bar{\sigma}_\rms^2{P_\rmt}\|\ep_m\|_\rmf^2 \|\eW\|_\rmf^2}.
\end{align}

\subsubsection{Updating Digital Precoder}
Thanks to the transformation from \eqref{P1} to \eqref{P2}, the subproblem of \eqref{P3} with respect to the digital beamforming matrix $\mathbf W_\rmd$ becomes unconstrained and convex, and can be formulated as
\begin{align}\label{P4W}
    \max_{\mathbf W_\rmd}\,\, &2\Re\{\rmtr(\mathbfeW\mathbf C_1^\H)\}\!\nonumber\\
    &-\!\rmtr(\mathbfeW(\mathbfeW)^\H\mathbf C_2).
\end{align}
By examining the first-order optimality condition, the optimal solution to subproblem \eqref{P4W} is given by the following closed-form expression:
\begin{equation}\label{update W}
    \mathbf W^\star=(\mathbf W_{\rma}^\H\mathbf W_{\rme}^\H\mathbf C_2\mathbf W_{\rme}\mathbf W_{\rma})^\dagger\mathbf W_{\rma}^\H\mathbf W_{\rme}^\H\mathbf C_1,
\end{equation}
where $(\cdot)^\dagger$ denotes the pseudo-inverse, which ensures a valid solution even when the matrix $\mathbf W_{\rma}^\H\mathbf W_{\rme}^\H\mathbf C_2\mathbf W_{\rme}\mathbf W_{\rma}$ is singular.

\subsubsection{Updating Analog Precoder}
\label{Sec:update analog}
Given other variables, the optimization problem for updating $\mathbf W_{\rma}$ can be expressed as
\begin{subequations}\label{P4F}
    \begin{align}
    \max_{\mathbf W_{\rma}}\,\, &2\Re\{\rmtr(\mathbf W_{\rma}\mathbf W_\rmd\mathbf C_1^\H\mathbf W_{\rme})\}\nonumber\\
    &-\!\rmtr(\mathbf W_{\rma}\mathbf W_\rmd\mathbf W_\rmd^\H\mathbf W_{\rma}^\H\mathbf W_{\rme}^\H\mathbf C_2\mathbf W_{\rme})\\
    \text{s.t.}\,\,&|[\mathbf W_{\rma}]_{i,j}|=1, \forall i\in\mathcal N_{\mathrm w},\forall j\in\mathcal N_{\mathrm{rf}}.
\end{align}
\end{subequations}
 In the following, we employ an efficient algorithm, namely the shifted generalized power iteration (SGPI), to solve problem \eqref{P4F}. Specifically, since the equality $\mathrm{tr}(\mathbf W_{\rma}\mathbf W_{\rma}^\H)=\Nw\Nrf$ holds for any feasible point, we can add the term $\lambda_1 \mathrm{tr}(\mathbf W_{\rma}\mathbf W_{\rma}^\H)$ to the objective function of \eqref{P4F}, where $\lambda_1\geq0$ is a predefined constant, without altering the optimal solution. With this transformation, we define an equivalent objective function $$
     g_1(\mathbf W_{\rma})\triangleq 2\Re\{\rmtr(\mathbf W_{\rma}^\H \mathbf L_1)\}+\rmtr(\mathbf W_{\rma}(\lambda_1\mathbf W_{\rma}^\H-\mathbf L_2\mathbf W_{\rma}^\H\mathbf L_3)),$$ where
\begin{align}\label{eq:L123}
 \mathbf L_1\triangleq \mathbf W_{\rme}^\H\mathbf C_1\mathbf W_\rmd^\H,\,\,\mathbf L_2\triangleq \mathbf W_\rmd\mathbf W_\rmd^\H,\,\, \mathbf L_3\triangleq \mathbf W_{\rme}^\H\mathbf C_2\mathbf W_{\rme}.
\end{align}
By choosing 
$\lambda_1=\widetilde{\omega}(\mathbf L_2)\times\widetilde{\omega}(\mathbf L_3)$, where $\widetilde{\omega}(\cdot)$ denotes the largest eigenvalue of a Hermitian matrix, the function $g_1(\mathbf W_{\rma})$ becomes convex. We then invoke the following lemma.

\begin{lemma}[\!\!\cite{boyd2004convex}]\label{lemma: linear lower bound}
Any convex function $g(\mathbf W_{\rma})$ is lower bounded by its first-order Taylor expansion at any feasible point $\mathbf W_{\rma}'$: 
\begin{equation}\label{linear lemma}
g(\mathbf W_{\rma})\geq g(\mathbf W_{\rma}')+\langle\nabla g(\mathbf W_{\rma})\big|_{\mathbf W_{\rma}=\mathbf W_{\rma}'},\mathbf W_{\rma}-\mathbf W_{\rma}'\rangle,
\end{equation}
where the equality holds if and only if $ \mathbf W_{\rma}=\mathbf W_{\rma}'$.
\end{lemma}

Base on Lemma \ref{lemma: linear lower bound}, by setting $\mathbf W_{\rma}'=\mathbf W_{\rma}^{[n]}$, a linear lower bound of problem \eqref{P4F} at iteration $n$ can be derived as
\begin{align}\label{P4F:re3}
     \max_{\mathbf W_{\rma}\in\mathcal C_{\rm A}}\,\, &\Re\{\rmtr(\mathbf W_{\rma}^\H \mathbf L_1)\}+\Re\{\rmtr(\mathbf W_{\rma}^\H(\lambda_1\mathbf W_{\rma}^{[n]}-\mathbf L_3\mathbf W_{\rma}^{[n]}\mathbf L_2))\},
\end{align}
whose optimal solution is given by
\begin{equation}\label{Update F}
	\mathbf W_{\rma}^\star= \exp\left(\rmj \angle\left(\mathbf L_1+\lambda_1\mathbf W_{\rma}^{[n]}-\mathbf L_3\mathbf W_{\rma}^{[n]}\mathbf L_2 \right)\right),
\end{equation}
with $\angle \mathbf X$ denoting the element-wise phase angle of $\mathbf X$. We summarize the SGPI algorithm for solving \eqref{P4F} in Algorithm \ref{alg: SGPI}.

\begin{algorithm}[t]
    \textbf{Input}: Compute $\mathbf L_1,\mathbf L_2$ and $\mathbf L_3$ based on \eqref{eq:L123}\;
\textbf{Initialize}: $n\leftarrow0, \mathbf{F}\in \mathcal{C}_{\rm A}$\;
Compute the shift parameter $\lambda_1=\widetilde{\omega}(\mathbf L_2)\times\widetilde{\omega}(\mathbf L_3)$\;
\Repeat{The objective function convergence}{

Update $ \mathbf W_{\rma}^{[n+1]} $ based on \eqref{Update F}\;
$ n\leftarrow n+1 $\;}
     
\caption{The SGPI algorithm for solving \eqref{P4F}}

\label{alg: SGPI}				
\end{algorithm}

\subsubsection{Updating DMA Precoder}
\label{Sec:update DMA}
Since $\mathbf W_{\rme}$ is a linear function of $\bm\Psi$, we directly optimize $\bm\Psi$. With all other variables fixed, \eqref{P4} can be solved with respect to $\bm\Psi$ in the following problem:
\begin{subequations}\label{P4M}
    \begin{align}
    \label{P4obj}\max_{\mathbf \Psi}\,\, &2\Re\{\rmtr(\mathbf W_{\rme}^\H\mathbf L_4)\}\!-\!\rmtr(\mathbf W_{\rme}^\H\mathbf C_2\mathbf W_{\rme}\mathbf L_5)\\
    \text{s.t.}\,\,&\mathbf W_{\rme}=\frac{1}{2}(\rmj\mathbf Q+\mathbf \Psi\circ\mathbf Q),\\
    \label{P4MC2}&|[\mathbf \Psi]_{m,n}|=
    \begin{cases}
        1,&m\in\mathcal N_{\rm r}, n=\lfloor \frac{m-1}{\Nu}\rfloor+1,\\
        0,  & \text{otherwise},
    \end{cases}
\end{align}
\end{subequations}
where 
\begin{equation}\label{eq:L45}
    \mathbf L_4\triangleq \mathbf C_1\mathbf W_\rmd^\H\mathbf W_{\rma}^\H ,\,\, \mathbf L_5\triangleq \mathbf W_{\rma}\mathbf W_\rmd\mathbf W_\rmd^\H\mathbf W_{\rma}^\H.
\end{equation}
This problem shares a similar structure with \eqref{P4F} and can thus be addressed using the proposed SGPI method. However, the partially connected architecture of the DMA imposes a sparse structure on the DMA beamforming matrix, where only $\Nr$ out of $\Nr\Nw$ elements are nonzero. This sparsity can be exploited to further reduce computational complexity.

By defining $\widetilde{\mathbf w}_\rme\triangleq \rmvec(\mathbf W_{\rme}), \widetilde{\mathbf l}\triangleq \rmvec(\mathbf L_4), \widetilde{\bm\psi}\triangleq \rmvec(\mathbf \Psi)$ and $\widetilde{\mathbf q}\triangleq \rmvec(\mathbf Q)$, problem \eqref{P4M} can be equivalently written as:
\begin{subequations}\label{P4Mvec}
    \begin{align}
    \max_{\widetilde{\bm\psi}}\,\, &2\Re\{\widetilde{\mathbf w}_\rme^\H\widetilde{\mathbf l}\}-\!\widetilde{\mathbf w}_\rme^\H(\mathbf L_5^\T\otimes\mathbf C_2)\widetilde{\mathbf w}_\rme\\
    \text{s.t.}\,\,&\widetilde{\mathbf w}_\rme=\frac{1}{2}\left(\rmj\widetilde{\mathbf q}+\widetilde{\bm\psi}\circ\widetilde{\mathbf q}\right),\\
    &|[\widetilde{\bm \psi}]_{(n-1)\Nr+m}|=
    \begin{cases}
        1,&m\in\mathcal N_{\rm r}, n=\lfloor \frac{m-1}{\Nu}\rfloor+1,\\
        0,  & \text{otherwise},
    \end{cases}
\end{align}
\end{subequations}
where $\otimes$ refers to the Kronecker product. To reduce the variable dimension in this vectorized formulation, we introduce the linear transformation:
\begin{equation}\label{Linear Trans}
    \widetilde{\mathbf w}_\rme=\mathbf R \mathbf w_{\rme},
\end{equation}
where $\mathbf w_{\rme}\triangleq[\mathbf w_{\rme 1}^\T,\ldots,\mathbf w_{\rme \Nw}^\T]^\T \in\mathbb C^{\Nr\times 1}$ collects all the nonzero elements of $\mathbf W_{\rme}$ and $\mathbf R\in\{0,1\}^{\Nr\Nw\times \Nr}$ is a binary transformation matrix that embeds $\mathbf w_{\rme}$ into the appropriate locations of $\widetilde{\mathbf w}_\rme$. Specifically, $\mathbf R$ is defined element-wise as
\begin{equation}
    [\mathbf R]_{m,n}=
    \begin{cases}
        1,&n\in\mathcal N_{\rm r}, m=\lfloor \frac{n-1}{\Nu}\rfloor\Nr+n,\\
        0,  & \text{otherwise},
    \end{cases}
\end{equation}
Substituting \eqref{Linear Trans} into \eqref{P4Mvec}, the problem reduces to
\begin{subequations}\label{P4Mvec2}
    \begin{align}
    \max_{\bm\psi}\,\, &2\Re\{\mathbf w_{\rme}^\H\mathbf R^\H\widetilde{\mathbf l}\}-\!\mathbf w_{\rme}^\H\mathbf R^\H(\mathbf L_5^\T\otimes\mathbf C_2)\mathbf R\mathbf w_{\rme}\\
   \label{P4Mvec2C1} \text{s.t.}\,\,&\mathbf w_{\rme}=\frac{1}{2}\left(\rmj\mathbf q+\bm\psi\circ\mathbf q\right),\\
    &|[\bm\psi]_m|=1, \forall m\in\mathcal N_r,
\end{align}
\end{subequations}
where $\bm\psi$ and $\mathbf q$ respectively collect the nonzero elements in $\bm\Psi$ and $\mathbf Q$ analogously to $\mathbf w_{\rme}$. 

Follow the same SGPI approach used to solve \eqref{P4F}, we apply it to \eqref{P4Mvec2} as well. To ensure the convexity of the objective function, the shift parameter $\lambda_2$ is chosen as
\begin{equation}\label{update lambda}
\lambda_2=\widetilde{\omega}(\mathbf R^\H(\mathbf L_5^\T\otimes \mathbf C_2)\mathbf R).    
\end{equation}
Let $g_2(\bm\psi)$ denote the objective function in \eqref{P4Mvec2}, its gradient with respect to $\bm\psi$ is given by
\begin{equation}\label{gradient function}
    \nabla g_2(\bm\psi)=2\left((\mathbf R^\H(\mathbf L_5^\T\otimes \mathbf C_2)\mathbf R)\mathbf w_{\rme}+\mathbf R^\H\rmvec(\mathbf L_4)\right)\circ \mathbf q^*.
\end{equation}
For a given feasible point $\bm\psi^{[n]}$ at iteration $n$, $\bm\psi$ admits the following clsoed-form solution:
\begin{equation}\label{Update psi}
\bm\psi^\star=\exp{\left(\rmj \angle \left(\lambda_2\bm\psi^{[n]}+ \nabla g_2(\bm\psi)\big|_{\bm\psi=\bm\psi^{[n]}} \right)\right)}.
\end{equation}
Then, the DMA beamformer can be reconstructed based on \eqref{Linear Trans} and \eqref{P4Mvec2C1}. We summarize the SGPI algorithm for solving problem \eqref{P4M} in Algorithm \ref{alg: SGPI for DMA}.
\begin{algorithm}[t]
    \textbf{Input}:  Compute $\mathbf L_4,\mathbf L_5$ and $\mathbf C_2$ based on \eqref{eq:L45} and \eqref{eq:C12}\;
\textbf{Initialize}: $n\leftarrow0, \mathbf{M}\in \mathcal{C}_{\rm D}$\;
Compute the shift parameter $\lambda_2$ using \eqref{update lambda}\;
Construct the gradient function in \eqref{gradient function}\; 
\Repeat{The objective function convergence}{

Update $ \bm\psi^{[n+1]} $ according to \eqref{Update F}\;
$ n\leftarrow n+1 $\;}
Construct DMA beamformer $\mathbf W_{\rme}$ by \eqref{Linear Trans} and \eqref{P4Mvec2C1}    
\caption{The SGPI algorithm for solving \eqref{P4M}}

\label{alg: SGPI for DMA}				
\end{algorithm}

\subsubsection{Updating Combiners}

The update procedure for the sensing combining matrices follows the same structure as the transmit beamforming updates. For brevity, we outline only the key steps. First, to facilitate the design of the digital, analog, and DMA combiners, we  rewrite problem \eqref{P4} into a compact form with respect to ${\mathbf P_\rmd,\mathbf P_{\rma},\mathbf P_{\rme}}$:
\begin{subequations}\label{P5}
    \begin{align} \max_{\bm\Omega,\bm\eta,\bm\beta}\,\, &2\Re\{\rmtr(\mathbf P_{\rme}\mathbf P_{\rma}\mathbf P_\rmd\mathbf C_3^\H)\}\!-\!\rmtr(\mathbf P_{\rme}\mathbf P_{\rma}\mathbf P_\rmd\mathbf C_4(\mathbf P_{\rme}\mathbf P_{\rma}\mathbf P_\rmd)^\H\mathbf C_5)\\
    \text{s.t.}\,\, 
    & \eqref{P1C1},\eqref{P1C3},
\end{align}
\end{subequations}
where 
\begin{subequations}\label{eq:C345}
    \begin{align}
        &\mathbf C_3\triangleq[\sqrt{1+\alpha_{\rms 1}}\mathbf G_1\eW\bm\theta_{\rms 1}^\H,\ldots,\sqrt{1+\alpha_{\rms M}}\mathbf G_M\eW\bm\theta_{\rms M}^\H],\\
        &\mathbf C_4\triangleq \mathrm{diag}(\|\bm\theta_{\rms 1}\|_\rmf^2,\ldots,\|\bm\theta_{\rms M}\|_\rmf^2), \\
        &\mathbf C_5\triangleq \sum_{j=1}^{M+C} \mathbf G_j\eW\eW^\H\mathbf G_j^\H+\sigma_\rms^2\|\eW\|_\rmf^2/P_\rmt\mathbf I.
    \end{align}
\end{subequations}

\noindent
The digital combiner is updated with
\begin{equation}\label{update P}
    \mathbf P_\rmd^\star=(\mathbf P_{\rma}^\H\mathbf P_{\rme}^\H\mathbf C_5\mathbf P_{\rme}\mathbf P_{\rma} )^\dagger \mathbf P_{\rma}^\H\mathbf P_{\rme}^\H\mathbf C_3 \mathbf C_4^{-1}.
\end{equation}

\noindent 
Let $\acute{\mathbf L}_1\triangleq \mathbf P_{\rme}^\H\mathbf C_3\mathbf P_\rmd^\H, \acute{\mathbf L}_2\triangleq \mathbf P_\rmd\mathbf C_4\mathbf P_\rmd^\H$, and $\acute{\mathbf L}_3\triangleq \mathbf P_{\rme}^\H\mathbf C_5\mathbf P_{\rme}$, the analog combiner is updated as
\begin{equation}
    \mathbf P_{\rma}^\star=\exp\left(\rmj \angle \left(\acute{\mathbf L}_1+\lambda_3 \mathbf P_{\rma}^{[n]}-\acute{\mathbf L}_3\mathbf P_{\rma}^{[n]}\acute{\mathbf L}_2\right)\right),
\end{equation}
with $\lambda_3=\widetilde{\omega}(\acute{\mathbf L}_2)\times\widetilde{\omega}(\acute{\mathbf L}_3)$. 

\noindent
Define
$\acute{\mathbf L}_4\triangleq \mathbf C_3\mathbf P_\rmd^\H\mathbf P_{\rma}^\H $
and $ \acute{\mathbf L}_5\triangleq \mathbf P_{\rma}\mathbf P_\rmd\mathbf C_4\mathbf P_\rmd^\H\mathbf P_{\rma}^\H$, the update of $\mathbf P_{\rme}$ follows exactly the same procedure as the DMA matrix update in Section~\ref{Sec:update DMA}.

\subsection{Overall Optimization Framework for Problem \eqref{P1} }

Combining the problem reformulation in Section \ref{Sec:Reformulation}, the FP method in Section \ref{sec:FP method}, and the BCD framework described in Section \ref{sec: BCD framework}, we summarize the complete optimization procedure for solving problem \eqref{P1} in Algorithm \ref{alg: FP-SGPI}. Starting from an initial nonzero feasible point $\{\mathbf W_{\rme}^{[0]},\mathbf W_{\rma}^{[0]},\mathbf W_\rmd^{[0]},\mathbf P_{\rme}^{[0]},\mathbf P_{\rma}^{[0]},\mathbf P_\rmd^{[0]}\}$, the auxiliary variables $\bm\eta^{[t]}$, $\bm\beta^{[t]}$ are first updated using \eqref{update auxilary variables}. Based on these variables, the matrices $\mathbf C_1^{[t]}$ and $\mathbf C_2^{[t]}$ are computed according to their definitions. Next, the DMA and analog beamforming matrices $\mathbf W_\rme^{[t]}$ and $\mathbf W_\rma^{[t]}$ are updated via the SGPI procedures described in Sections \ref{Sec:update DMA} and \ref{Sec:update analog}, respectively, followed by the closed-form update of the digital beamformer $\mathbf W_\rmd^{[t]}$ in \eqref{update W}. The update of the receive combining matrices mirrors that of the transmit side. Specifically, we compute $\mathbf C_3^{[t]},\mathbf C_4^{[t]}$, and $\mathbf C_5^{[t]}$, apply SGPI to update the DMA and analog combiners, and then update the digital combiner in closed form. These steps are iteratively executed until the objective value in \eqref{P1} converges. Finally, after convergence, the digital beamformer $\mathbf W_\rmd^{[t]}$ is rescaled to ensure satisfaction of the total transmit power constraint.

\begin{algorithm}[t]
	\textbf{Initialize}: $t\leftarrow0$, $\{\mathbf W_{\rme}^{[0]},\mathbf W_{\rma}^{[0]},\mathbf W_\rmd^{[0]},\mathbf P_{\rme}^{[0]},\mathbf P_{\rma}^{[0]},\mathbf P_\rmd^{[0]}\}$\;
	\Repeat{convergence of objective value in \eqref{P1}}{
		$t\leftarrow t+1$\;
		
        Update $\bm{\eta}^{[t]} $ and $ \bm{\beta}^{[t]}$ using \eqref{update auxilary variables}\;
        Compute $\mathbf{C}_1^{[t]}$ and $\mathbf{C}_2^{[t]}$ using \eqref{eq:C12}\;
   
        Update $\mathbf W_{\rme}^{[t]}$ using Algorithm \ref{alg: SGPI for DMA}\;
        Update $\mathbf{W}_\rma^{[t]}$ using Algorithm \ref{alg: SGPI}\;
        Update $\mathbf{W}_\rmd^{[t]}$ using the expression in \eqref{update W}\;
 Compute $\mathbf{C}_3^{[t]}, \mathbf C_4^{[t]}$ and $\mathbf{C}_5^{[t]}$ using \eqref{eq:C345}\;
   
        Update $\mathbf P_\rme^{[t]}$ using Algorithm \ref{alg: SGPI for DMA}\;
        Update $\mathbf{P}_\rma^{[t]}$ using Algorithm \ref{alg: SGPI}\;
        Update $\mathbf{P}_\rmd^{[t]}$ using the expression in \eqref{update P}\;
 	
	}	
     Obtain $\mathbf{W}_\rmd^{[t]}$ using \eqref{scale W}.
	\caption{Proposed FP-SGPI Algorithm for Solving Problem~\eqref{P1}}
	\label{alg: FP-SGPI}				
\end{algorithm}

\subsection{Convergence Analysis}
\label{sec:convergence}
The convergence of the proposed algorithm follows from standard properties of the SGPI and BCD frameworks. Specifically, for each subproblem (e.g., \eqref{P4F}), the SGPI update can be interpreted as a shifted power iteration applied to a complex matrix. With a properly chosen positive shift, each iteration yields a non-decreasing objective value and converges to a stationary point of the subproblem.

Based on this property, the overall FP-SGPI algorithm produces a sequence of objective values that is monotonically non-decreasing. Moreover, since each block is either updated in closed form or via SGPI, the algorithm falls into the standard BCD framework, which guarantees convergence to a stationary point of problem \eqref{P2} \cite{xu2013block}. By leveraging the equivalence in Proposition~\ref{pro:equivalence}, the obtained solution can be further scaled to yield a stationary point of the original problem \eqref{P1}.

\subsection{Computational Complexity Analysis}

The computational complexity of Algorithm~\ref{alg: FP-SGPI} is primarily attributed to matrix multiplications. In each iteration, the complexity of updating the auxiliary variables $\bm\eta^{[t]}, \bm\beta^{[t]}, \mathbf C_1^{[t]}$, and $\mathbf C_2^{[t]}$ is $\mathcal{O}(K^2\Nr + K\Nr^2 + K\Nw(\Nrf + \Nr))$. The update of the DMA beamforming matrix via the SGPI algorithm incurs a complexity of $\mathcal{O}(I_1\Nr^2)$, where $I_1$ is the number of SGPI iterations. Similarly, updating the analog beamforming matrix requires $\mathcal{O}(I_2(\Nrf^2\Nw + \Nrf\Nw^2))$ operations, with $I_2$ denoting the associated SGPI iteration count. The update of the digital beamforming matrix involves matrix inversion and multiplication operations, resulting in a complexity of $\mathcal{O}(\Nrf\Nw\Nr + K\Nrf\Nr + \Nrf\Nr^2 + \Nrf^3)$. 

Since the receiver adopts the same tri-HBF architecture, the updates of the DMA combiner, analog combiner, and digital combiner contribute a complexity of the same order as their transmit-side counterparts. Thus, the complexity of the complete tri-HBF transceiver scales as roughly twice that of the beamformer alone, without altering the dominant order terms. 

Combining these components, the overall complexity of Algorithm~\ref{alg: FP-SGPI} over $I_4$ outer iterations is $\mathcal{O}\big(I_4(K^2\Nr + K\Nr^2 + I_1\Nr^2 + I_2(\Nrf^2\Nw + \Nrf\Nw^2) + \Nrf\Nw\Nr + K\Nrf\Nr + \Nrf\Nr^2 + \Nrf^3)\big)$. Since $\Nr$ is typically much larger than the other parameters ($K$, $\Nrf$, and $\Nw$) in large-scale systems, the overall computational complexity can be approximated by $\mathcal{O}(\bar{I}N_{\mathrm{r}}^2)$, where $\bar{I} \triangleq I_4(K+I_1+N_{\mathrm{RF}})$. This indicates that the proposed method incurs a reasonable computational complexity for large-scale implementations.

\section{Simulation Results}\label{simulation}\label{Sec:simulation}

In this section, we evaluate the convergence behavior and performance of the proposed algorithms. Unless otherwise specified, the simulation parameters are set as follows: $\Nu=16$, $\Nw=8$, $\Nrf=4$, $K=4$, $C=2$ and $M=3$. The transmit power is fixed at $P_{\mathrm t}=10$ dBm, while the noise variances are $\sigma_{k}^2=0$ dBm, $\forall k$. We consider a DMA-based BS operating at a carrier frequency of $f_c=28$ GHz (corresponding to a wavelength of $\lambda_c=1.07$ cm). The inter-waveguide spacing is set to $\lambda_c/2$, while the spacing between adjacent radiating elements is $\lambda_c/5$~\cite{zhang2022beam}. The attenuation coefficient and wavenumber are set to $\vartheta=0.6~\text{m}^{-1}$ and $\varpi=827.67~\text{m}^{-1}$, respectively~\cite{chen2025energy}. Thus, the DMA is modeled as a uniform planar array (UPA) with $\Nw\times\Nu$ elements, where $\Nw$ waveguides and each waveguide contains $\Nu$ radiating elements arranged along the horizontal and vertical dimensions. The UPA steering vector $\mathbf a(\theta,\phi)$ is modeled as $\mathbf a(\theta,\phi)=\mathbf a_{\mathrm u}(\theta,\phi)\otimes \mathbf a_{\mathrm w}(\phi)$ \cite{nhan2022UPA}, where 
\begin{align*}
&\mathbf a_{\mathrm u}(\theta,\phi)={1}/{\sqrt{N_{\mathrm u}}}\left[1,e^{j\pi\sin{\theta}\sin{\phi}},\ldots,e^{j\pi(N_{\mathrm u}-1)\sin{\theta}\sin{\phi}}\right]^\T,\\
&\mathbf a_{\mathrm w}(\phi)={1}/{\sqrt{N_{\mathrm{w}}}}\left[ 1,e^{j\pi\cos{\phi}},\ldots,e^{j\pi(N_{\mathrm w}-1)\cos{\phi}}\right]^\T.    
\end{align*}
 The number of propagation paths is set to $L=10$. The azimuth angles $\theta_m$ are independently drawn from the uniform distribution $\mathcal U(-\pi/3,\pi/3)$, while the elevation angles $\phi_m$ are sampled from $\mathcal U(\pi/6,5\pi/6)$. The convergence tolerance for Algorithm~\ref{alg: FP-SGPI} is set to $10^{-4}$. All presented results are obtained by averaging over 100 independent channel realizations.

 The proposed algorithm is initialized with a feasible point. Specifically, the phases of the analog and DMA beamformers and combiners are independently drawn from the uniform distribution $\mathcal U[0,2\pi)$. The digital beamformer is initialized using regularized zero-forcing (RZF), while the digital combiner is obtained by maximizing each $\gamma_{\rm s m}$ with fixed analog and DMA combiners. The latter reduces to a generalized eigenvalue problem, which can be efficiently solved using standard methods.

\subsection{Power Consumption Model}

We adopt the BS power model
\begin{equation}
    P_{\text{tot}}^{x}=\rho_\text{a}^{-1} P_\text{t} + P_{\text{BS}} + P_{\text{dyn}}^{x},
\end{equation}
where $\rho_\text{a}$ is the PA efficiency, $P_{\text{BS}}$ is the static circuit power, and $P_{\text{dyn}}^{x}$ denotes the dynamic hardware power of architecture $x$.

In the DMA architecture, adjacent waveguides are spaced by half a wavelength, which is consistent with traditional arrays. 
However, the radiating elements along each waveguide can be placed with a smaller inter-element spacing, allowing more elements to be deployed within the same physical aperture. To enable fair comparisons with conventional antenna arrays, we consider two comparison strategies. \textbf{SA (same aperture)} refers to a half-wavelength spaced array whose physical aperture matches that of the DMA-based array. 
The corresponding number of antenna elements in vertical is
$N_{\mathrm u}^{\text{SA}} = \lfloor 2(\Nu-1)/5 \rfloor + 1$. \textbf{SN (same number)} uses a half-wavelength spaced array with the same number of antenna elements as the DMA-based array, i.e., $N_{\mathrm u}^{\text{SN}}=\Nu$, which results in a larger physical aperture.

The dynamic power consumption is summarized below.
\begin{itemize}
\item \textit{Tri-HBF (THB)}: $\Nrf$ RF chains and $\Nrf\Nw$ phase shifters,
\begin{equation}
    P_{\text{dyn}}^{\text{THB}}= \Nrf P_{\mathrm{RF}} + \Nrf \Nw P_{\mathrm{PS}}.
\end{equation}

\item \textit{DMA-HBF}: $\Nw$ waveguides, each driven by one RF chain (no phase shifters),
\begin{equation}
    P_{\text{dyn}}^{\text{DMA-HBF}}= \Nw P_{\mathrm{RF}}.
\end{equation}

\item \textit{Fully-connected (FC)-HBF}~\cite{yu2016alternating}: $\Nrf$ RF chains and $\Nrf\Nw\Nu$ phase shifters,
\begin{align}
    P_{\text{dyn}}^{\text{FC-HBF SA}}&= \Nrf P_{\mathrm{RF}} + \Nrf \Nw N_{\mathrm u}^{\text{SA}} P_{\mathrm{PS}},\\
    P_{\text{dyn}}^{\text{FC-HBF SN}}&= \Nrf P_{\mathrm{RF}} + \Nrf \Nw N_{\mathrm u} P_{\mathrm{PS}}.
\end{align}

\item \textit{Sub-connected (SC)-HBF}~\cite{yu2016alternating}: $\Nrf$ RF chains and $\Nw\Nu$ phase shifters,
\begin{align}
    P_{\text{dyn}}^{\text{SC-HBF SA}}&= \Nrf P_{\mathrm{RF}} + \Nw N_{\mathrm u}^{\text{SA}} P_{\mathrm{PS}},\\
    P_{\text{dyn}}^{\text{SC-HBF SN}}&= \Nrf P_{\mathrm{RF}} + \Nw N_{\mathrm u} P_{\mathrm{PS}}.
\end{align}

\item \textit{Fully-digital (FD)}: one RF chain per antenna,
\begin{align}
    P_{\text{dyn}}^{\text{FD SA}}&=  \Nw N_{\mathrm u}^{\text{SA}} P_{\mathrm{RF}},\\
    P_{\text{dyn}}^{\text{FD SN}}&=  \Nw N_{\mathrm u} P_{\mathrm{RF}}.
\end{align}
\end{itemize}

We set $\rho_{\text{a}}=0.3$, $P_{\text{BS}}=40$ dBm, $P_{\mathrm{RF}}=30$ dBm, and $P_{\mathrm{PS}}=30$ mW~\cite{you2023energy}.

\noindent The communications EE is defined as
\begin{equation}
\text{EE}_{\rmc}=\frac{\sum_{k=1}^K \log\left(1+\gamma_{{\rm c}k}\right)}{P_{\text{tot}}}.
\end{equation}
In contrast, sensing EE in ISAC systems does not yet have a unified definition. Prior works such as \cite{zou2024energy}
define sensing EE as the ratio between the inverse CRLB and the total power consumption. In this work, we define sensing EE as the sensing MI per unit power, i.e., 
 \begin{equation}
\text{EE}_\rms=\frac{\sum_{m=1}^M \log(1+\gamma_{\rms m})}{P_{\text{tot}}}.
\end{equation}
This definition mirrors the communications EE and provides a unified interpretation: both metrics quantify how much information is obtained per joule of consumed energy.

\begin{figure}[t!]
	\centering
	\subfigure[Convergence versus iteration index]{\label{fig:conv_iter}
	\includegraphics[width=0.8\linewidth]{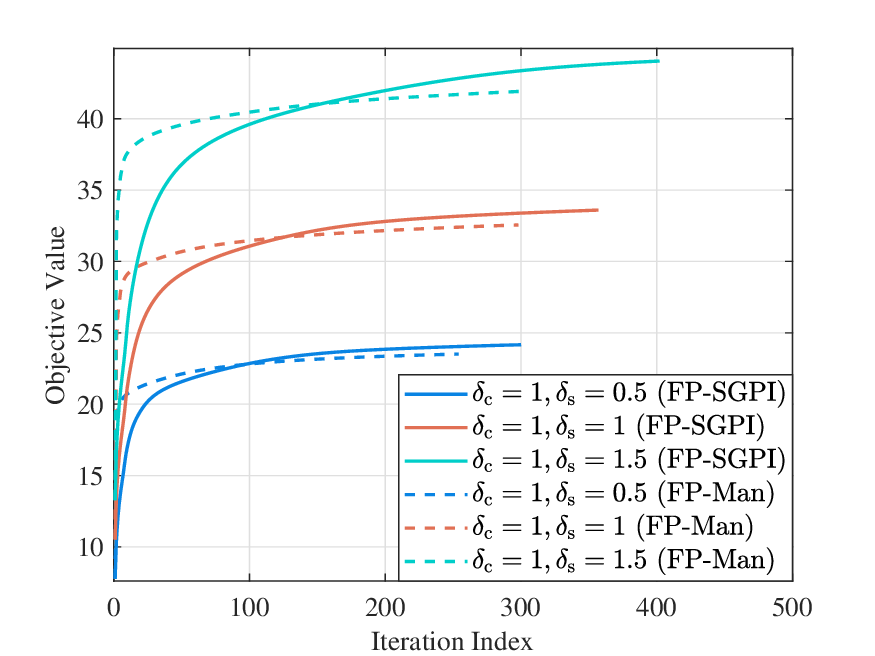}
}\\[1ex]
\subfigure[Convergence versus CPU time]{\label{fig:conv_time}
	\includegraphics[width=0.8\linewidth]{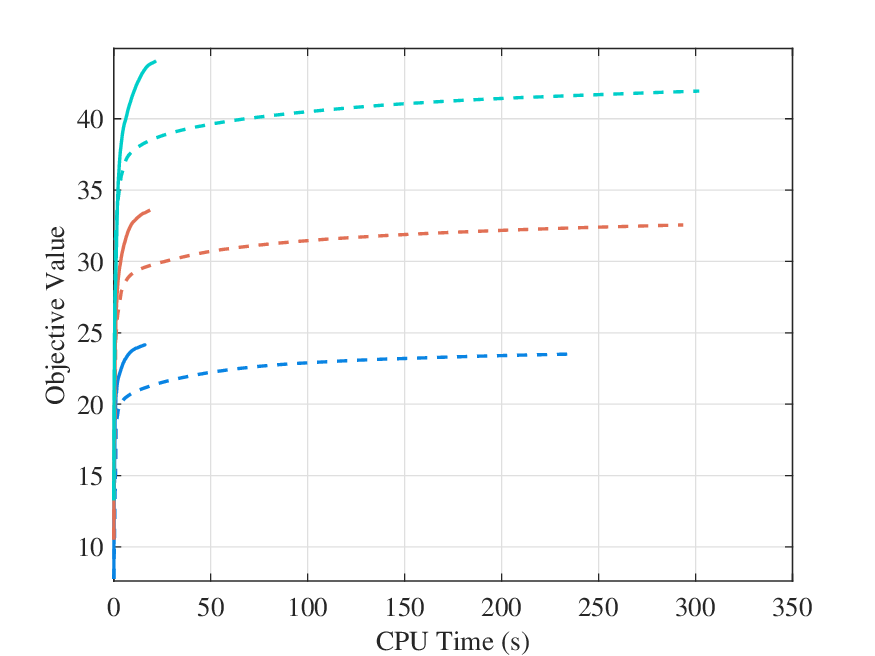}
}
	\caption{Convergence behavior of the FP-SGPI and FP-Man algorithms under different weight settings. 
(a) Objective value versus iteration index. 
(b) Objective value versus CPU time. }
	\label{fig:convergence}
\end{figure}

\begin{figure}[t!]
	\centering
	\subfigure[Communications sum rate vs. sensing sum MI.]{\label{fig2a:SR_vs_SMI}
	\includegraphics[width=0.8\linewidth]{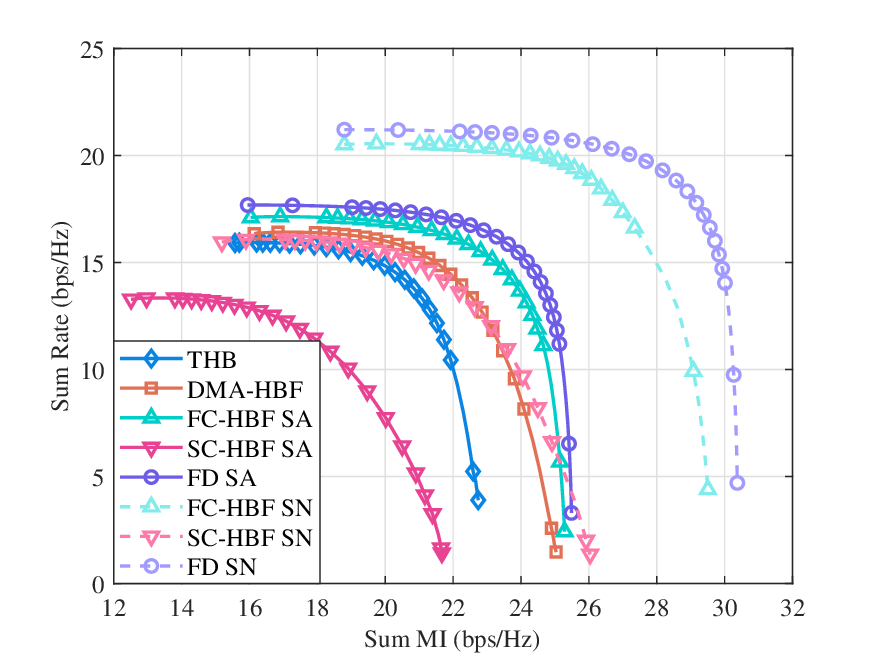}
}\\[1ex]
\subfigure[Communications EE vs. sensing EE.]{\label{fig2b:CEE_vs_SEE}
		\includegraphics[width=0.8\linewidth]{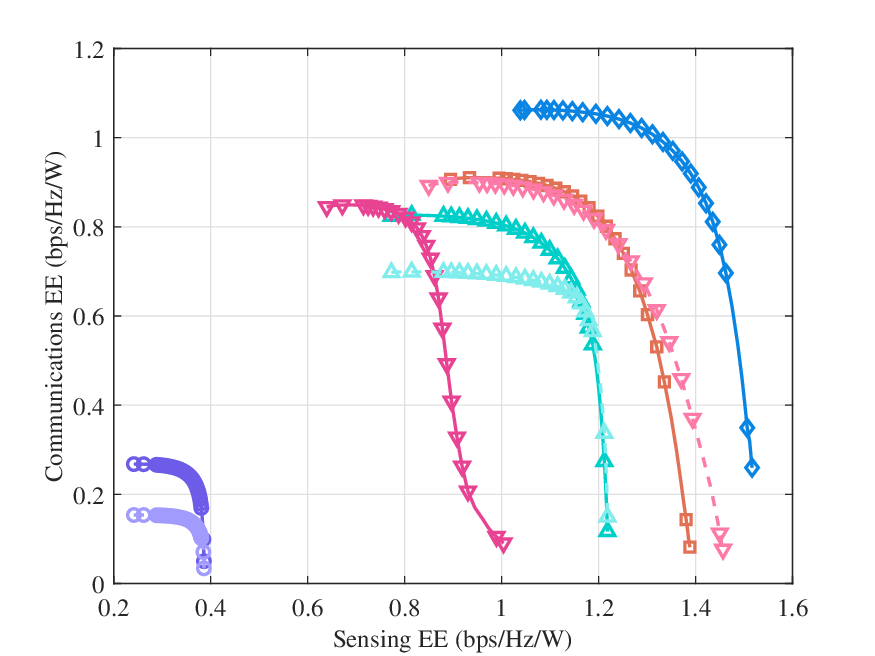}
}
	\caption{Tradeoff performance of different architectures: 
    (a) sum rate vs. sum MI, 
    (b) Communications EE vs. sensing EE.}
	\label{fig2:SR_EE_tradeoff}
\end{figure}

\subsection{Convergence}

Fig.~\ref{fig:convergence} illustrates the convergence behavior of the proposed FP-SGPI algorithm under different choices of weight coefficients, with FP-Man included as a benchmark for comparison. It is observed that, in all cases, the proposed algorithm exhibits a monotonic increase in the objective value and converges within a limited number of iterations, demonstrating its effectiveness and stability. Moreover, while FP-Man typically converges in fewer iterations, the proposed FP-SGPI algorithm achieves faster convergence in terms of CPU time due to its significantly lower per-iteration computational complexity.

The FP-Man benchmark is implemented by replacing the SGPI-based update in the proposed framework with a manifold optimization step. Specifically, each subproblem is solved on the corresponding manifold using the Manopt toolbox, where a conjugate gradient method with line search is adopted. As a result, each iteration of FP-Man incurs a relatively high computational cost. In contrast, the proposed FP-SGPI algorithm leverages a closed-form shift parameter to update the variables, thereby avoiding line search and significantly reducing the per-iteration complexity. 

Finally, the results also demonstrate the flexibility of the proposed algorithm in balancing sensing and communication performance by appropriately tuning the weight coefficients.
\subsection{Basic Performance Trade-off}
Fig.~\ref{fig2:SR_EE_tradeoff} illustrates the fundamental tradeoff between communications and sensing performance under various transceiver architectures. In Fig.~\ref{fig2a:SR_vs_SMI}, the communications sum rate is plotted against the sensing sum MI. As the sum MI increases, the achievable sum rate decreases because more transmit resources are devoted to sensing. The THB architecture achieves a larger tradeoff region than the SC-HBF SA scheme, but a smaller region compared with the fully digital, FC-HBF, and DMA-HBF baselines. Notably, the performance gap on the sensing side is more pronounced than on the communications side, this is because both the transmitter and receiver use the THB structure for sensing, which amplifies the cumulative performance loss. The reduced performance of THB primarily stems from its physical constraints, including waveguide attenuation, DMA beamforming limitations, and the imposed subarray structure. However, THB achieves superior EE. This advantage is highlighted in Fig.~\ref{fig2b:CEE_vs_SEE}, which shows the communications EE versus the sensing EE. It is seen that, the THB architecture exhibits the largest EE tradeoff region among all considered designs, outperforming both conventional antenna-array-based HBF and the DMA-HBF baseline. By incorporating a phase-shifter network to reduce the number of required RF chains in the DMA-HBF architecture, the tri-HBF scheme achieves further gains in overall EE.

\subsection{Impact of Transmit Power}\label{Sec:transmit power}

\begin{figure}[t!]
    \centering
    \hspace{-0.7 cm}
      \subfigure[Communications sum rate]{
        \label{fig3a:SR_vs_TP}
        \includegraphics[width=0.52\linewidth]{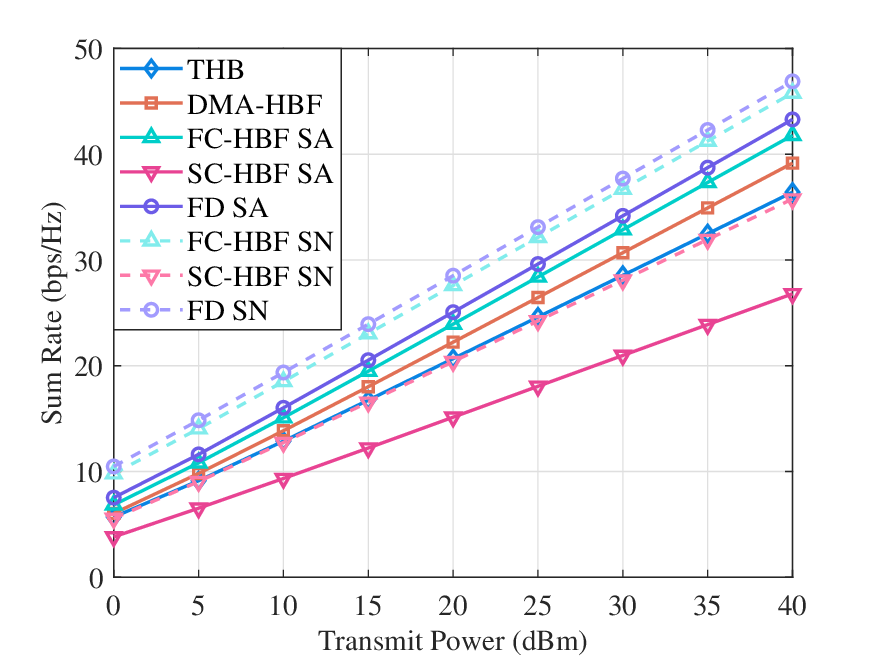}
    }\hspace{-0.7 cm}
    \subfigure[Sensing sum MI]{
        \label{fig3b:SMI_vs_TP}
        \includegraphics[width=0.52\linewidth]{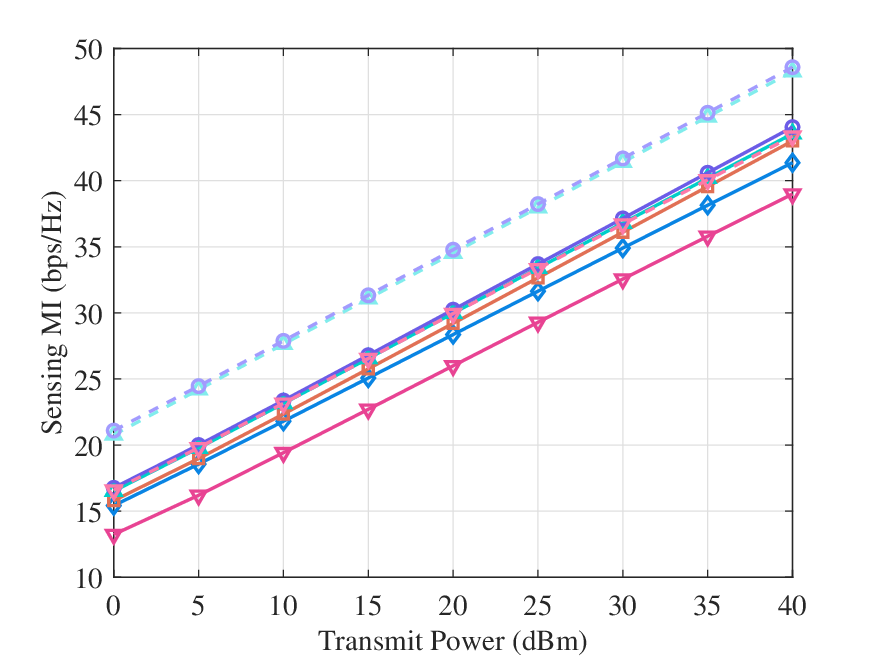}
    }

   \hspace{-0.7 cm} \subfigure[Communications EE]{
        \label{fig3c:CEE_vs_TP}
        \includegraphics[width=0.52\linewidth]{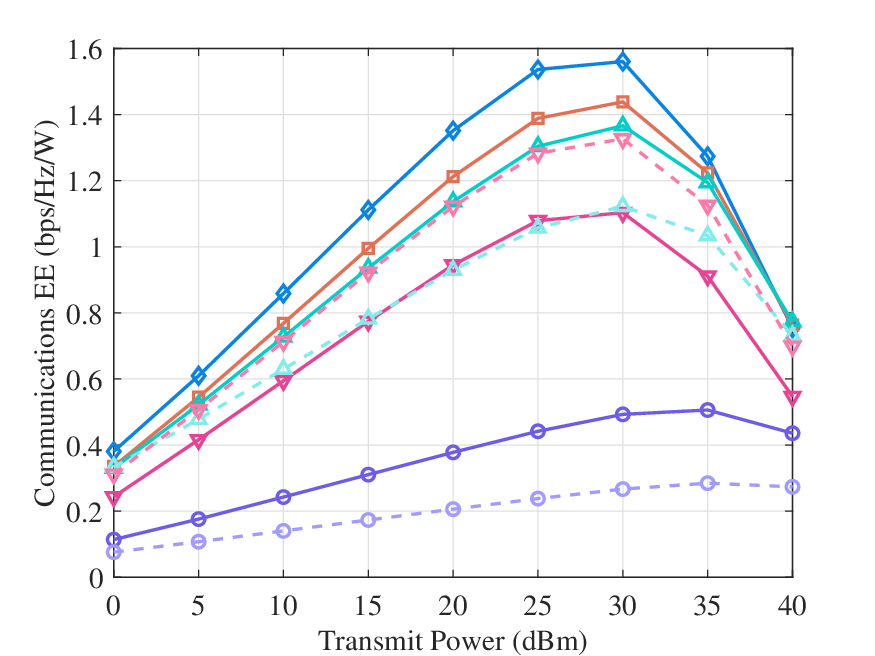}
    }\hspace{-0.7 cm}
    \subfigure[Sensing EE]{
        \label{fig3d:SEE_vs_TP}
        \includegraphics[width=0.52\linewidth]{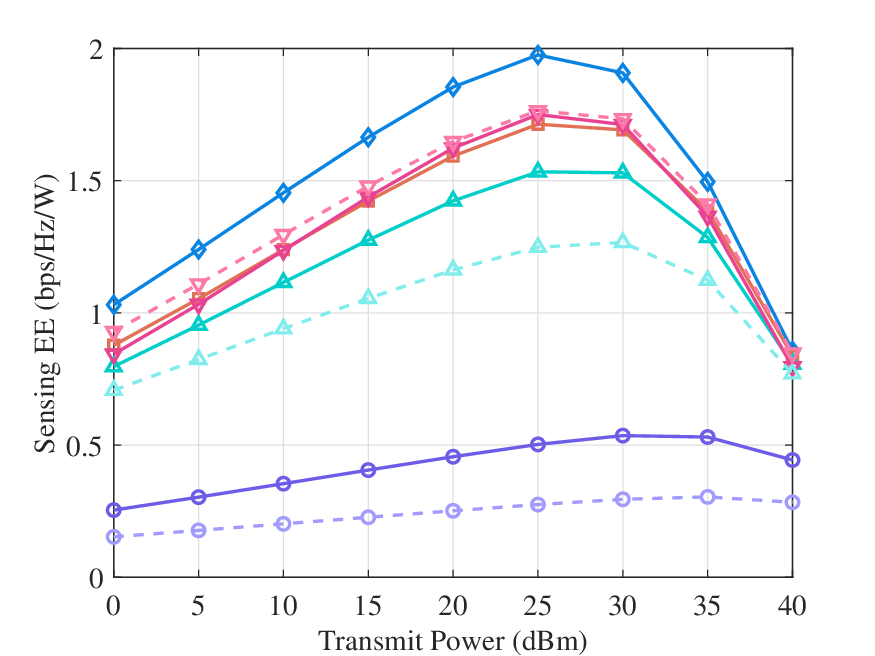}
    }
    \caption{Performance of different transceiver architectures under varying transmit power:
    (a) sum rate, 
    (b) sum MI,
    (c) communications EE, and 
    (d) sensing EE.}
    \label{fig3:performance_vs_TP}
\end{figure}

Fig.~\ref{fig3:performance_vs_TP} illustrates the performance of various transceiver architectures as the transmit power increases from 0 dBm to 40 dBm, with the weight coefficients set to $\delta_\rmc=\delta_\rms=1$. As expected, both the sum rate and the sum MI increase monotonically with transmit power. In Fig.~\ref{fig3a:SR_vs_TP}, the THB architectures exhibit communications performance comparable to the SC-HBF SN scheme, outperforming the SC-HBF SA but remaining inferior to the FD SA/SN and FC-HBF SA/SN counterparts. A similar ordering is observed for the sum MI in Fig.~\ref{fig3b:SMI_vs_TP}, except that the SC-HBF SN achieves slightly higher sum MI than the THB architectures. These observations are consistent with the trends reported in Fig.~\ref{fig2:SR_EE_tradeoff}. However, the EE in Fig.~\ref{fig3c:CEE_vs_TP} and Fig.~\ref{fig3d:SEE_vs_TP} show that the proposed THB architecture attains the highest EE, primarily due to its substantially lower dynamic power consumption.

\subsection{Impact of the Number of Elements per Waveguide}

\begin{figure}[t!]
    \centering
    \hspace{-0.7 cm}
      \subfigure[Communications sum rate]{
        \label{fig4a:SR_vs_Nu}
        \includegraphics[width=0.52\linewidth]{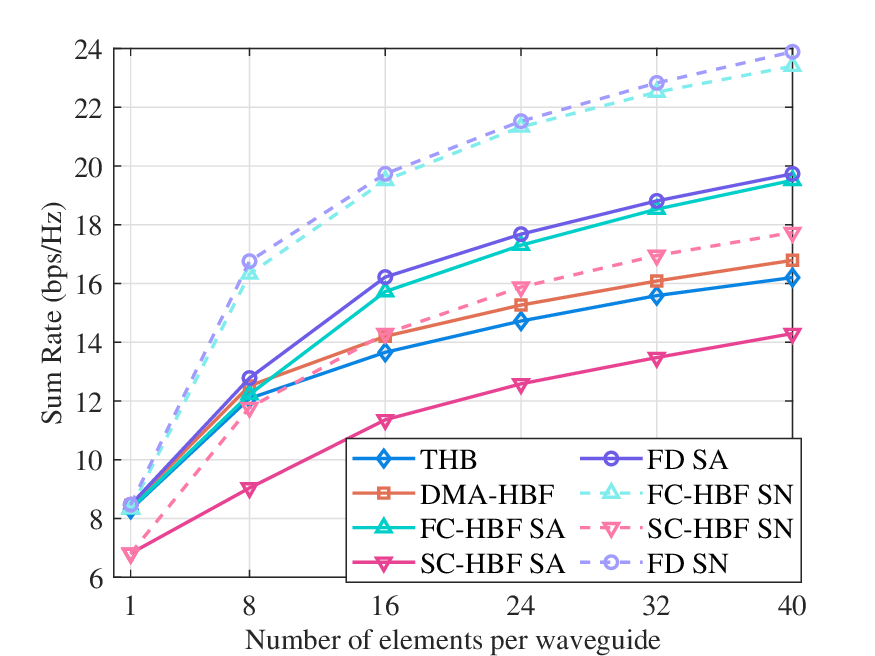}
    }\hspace{-0.7 cm}
    \subfigure[Sensing sum MI]{
        \label{fig4b:SMI_vs_Nu}
        \includegraphics[width=0.52\linewidth]{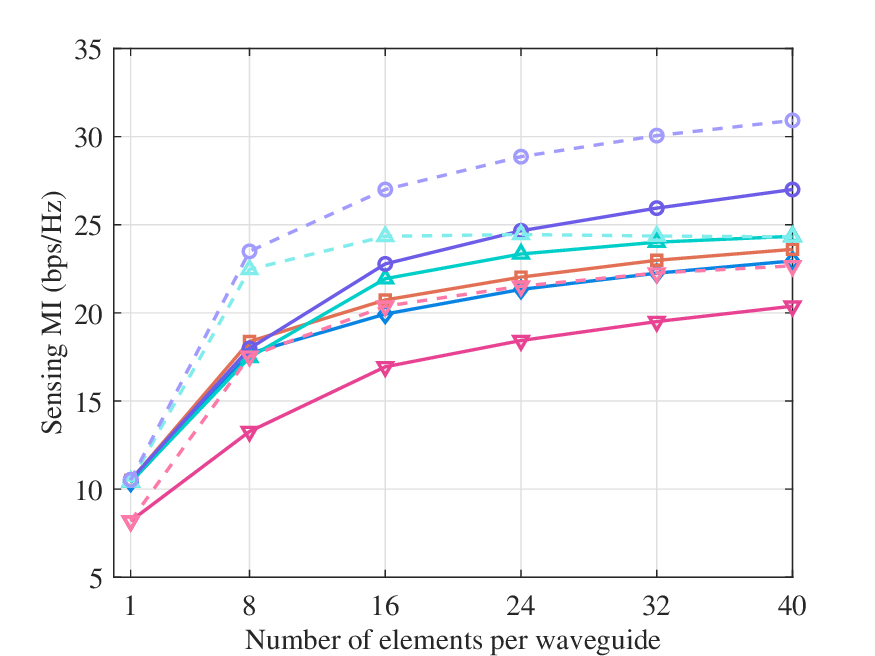}
    }

   \hspace{-0.7 cm} \subfigure[Communications EE]{
        \label{fig4c:CEE_vs_Nu}
        \includegraphics[width=0.52\linewidth]{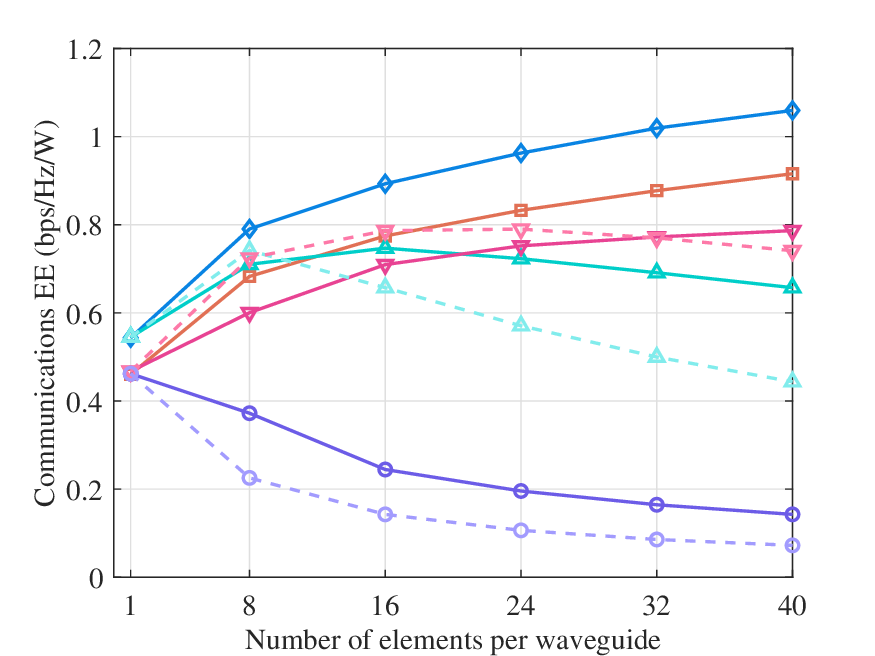}
    }\hspace{-0.7 cm}
    \subfigure[Sensing EE]{
        \label{fig4d:SEE_vs_Nu}
        \includegraphics[width=0.52\linewidth]{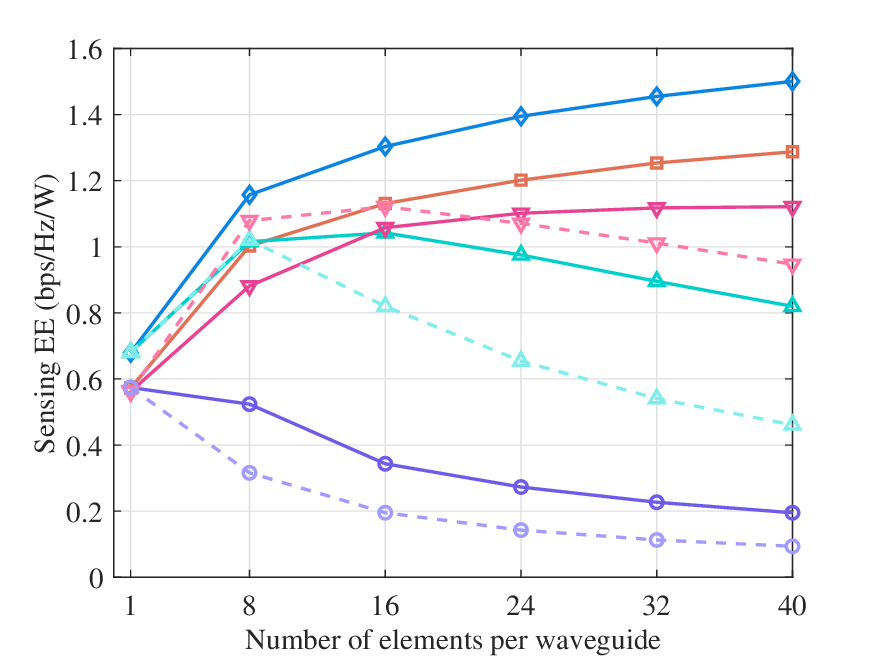}
    }
    \caption{Performance of different transceiver architectures under varying number of radiating elements per waveguide:
    (a) sum rate, 
    (b) sum MI,
    (c) communications EE, and 
    (d) sensing EE.}
    \label{fig4:performance_vs_Nu}
\end{figure}

Fig.~\ref{fig4:performance_vs_Nu} illustrates the performance of different transceiver architectures as the number of radiating elements per waveguide with $\Nu\in[1,8,16,24,32,40]$. The weight coefficients are set as $\delta_\rmc = \delta_\rms = 1$, and $P_\rmt = 10$ dBm. As shown in Fig.~\ref{fig4a:SR_vs_Nu}, the sum rate increases with $\Nu$ for all architectures. The DMA-HBF architecture achieves slightly higher performance than the proposed THB scheme, and both inferior to the FD SA/SN, FC-HBF SA/SN, and SC-HBF SN architectures in terms of achievable sum rate. Notably, the sum rate of SC-HBF SN is initially lower than that of THB but surpasses it as $\Nu$ increases, whereas the SC-HBF SA configuration remains consistently inferior to THB across the entire range of $\Nu$. Fig.~\ref{fig4b:SMI_vs_Nu} exhibits a similar trend to Fig.~\ref{fig4a:SR_vs_Nu}, except that the SC-HBF SN architecture achieves a sum MI comparable to that of the THB scheme. Fig.~\ref{fig4c:CEE_vs_Nu} and Fig.~\ref{fig4d:SEE_vs_Nu} demonstrate that the THB scheme consistently achieves the highest communications and sensing EE among all considered schemes. This superiority stems from its ability to strike a favorable balance between performance enhancement and hardware efficiency. In particular, as $\Nu$ increases, both the sum rate and sum MI grow substantially while the additional DMA radiating elements do not introduce extra hardware power consumption, thereby improving the numerator of the EE metric without affecting the denominator.
Since adding radiating elements to the waveguide does not incur additional power consumption, both the THB and DMA-HBF architectures exhibit increasing EE as $\Nu$ grows. In contrast, the FD baselines show decreasing EE, while the other HBF baselines first improve and then decline as $\Nu$ increases.
\subsection{Impact of the Number of Communications Users}

\begin{figure*}[t!]
	\centering
	\subfigure[$K=2$.]{\label{fig:SR_vs_SP_K1M3}
	\includegraphics[width=0.32\linewidth]{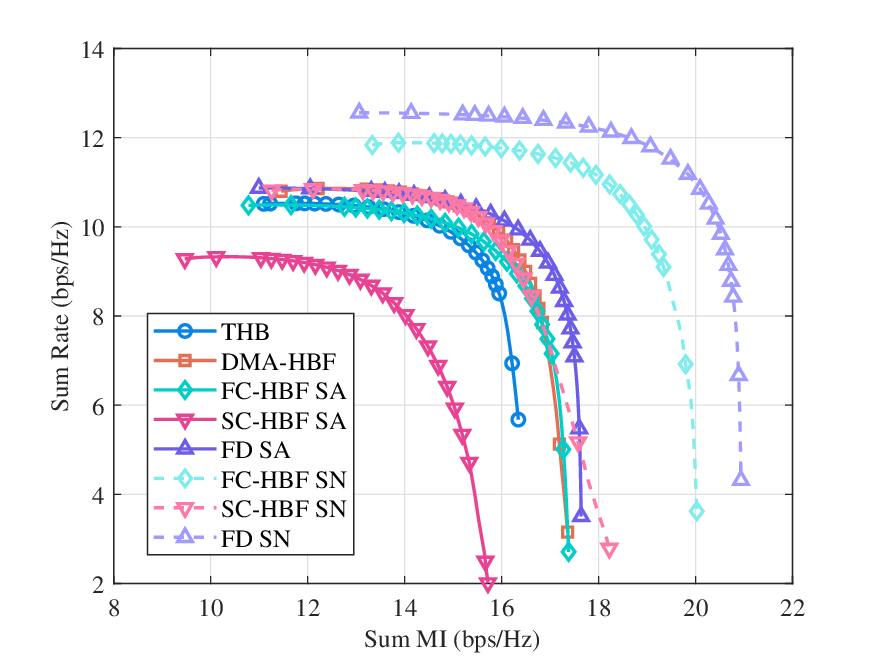}
}\hspace{-0.5 cm}
\subfigure[$K=3$.]{\label{fig:SR_vs_SP_K2M3 }
		\includegraphics[width=0.32\linewidth]{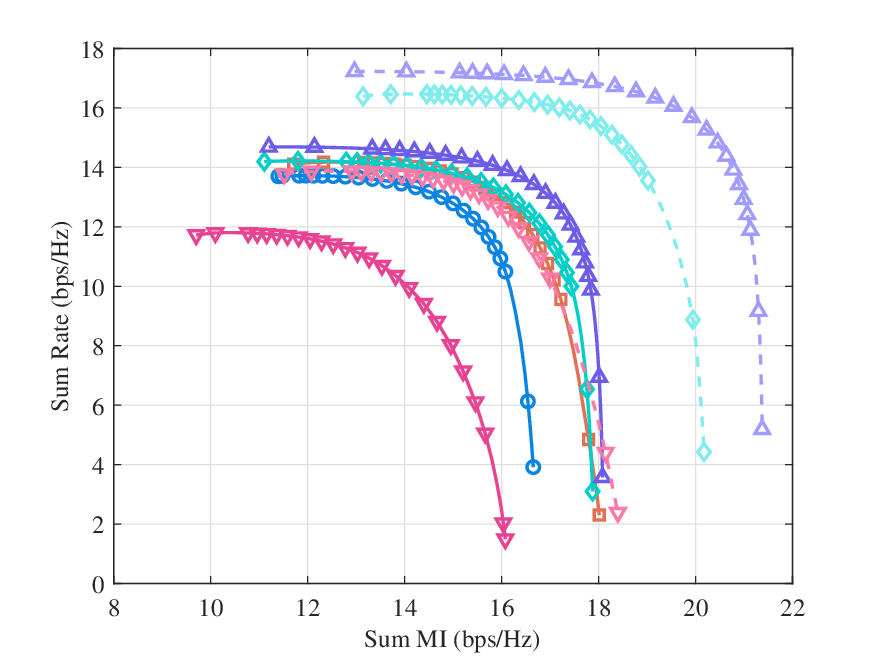}
}\hspace{-0.5 cm}
\subfigure[$K=4$.]{\label{fig:SR_vs_SP_K3M3}
		\includegraphics[width=0.32\linewidth]{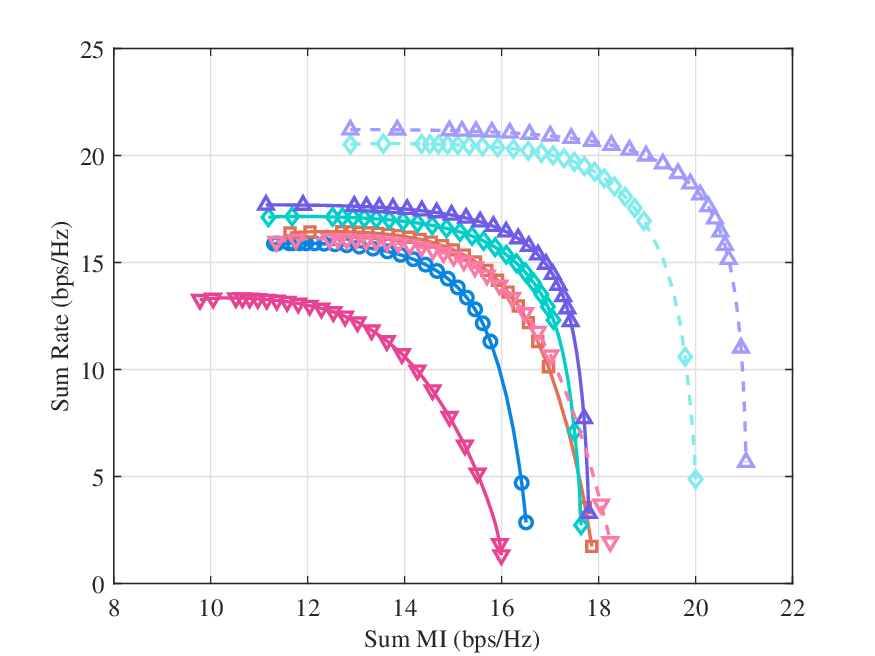}
}
	\caption{Comparison of different transceiver architectures with varying numbers of communications users ($K=1,2,3$) while fixing $M=2$. The tradeoff between sum rate and sum MI is illustrated for (a) $K=2$, (b) $K=3$, and (c) $K=4$.}

	\label{fig:performance_vs_K}
\end{figure*}

Fig.~\ref{fig:performance_vs_K} illustrates the performance of different transceiver architectures as the number of communications users increases from 2 to 4, with the number of sensing targets fixed at $M=2$. The results show that increasing the number of users enlarges the communications-only region, while the sensing-only region remains nearly unaffected. Interestingly, the communications-only region of the DMA-HBF system is larger than that of the FC-HBF SA architectures when $K=2$, becomes comparable when $K=3$, and falls below when $K=4$. This behavior can be explained by the fact that DMA-based architectures can accommodate more radiating elements within the same physical aperture, thereby achieving higher spatial gain than conventional antenna array-based systems. For a small number of users, this additional spatial gain translates into higher sum rates. As the number of users increases, however, multiuser interference becomes more significant, and the limited beamforming flexibility of DMA-based systems leads to performance degradation.

\subsection{Impact of the Number of Sensing Targets}

\begin{figure*}[t!]
	\centering
	\subfigure[M=1.]{\label{fig5a:SR_vs_SMI_K4M1}
	\includegraphics[width=0.32\linewidth]{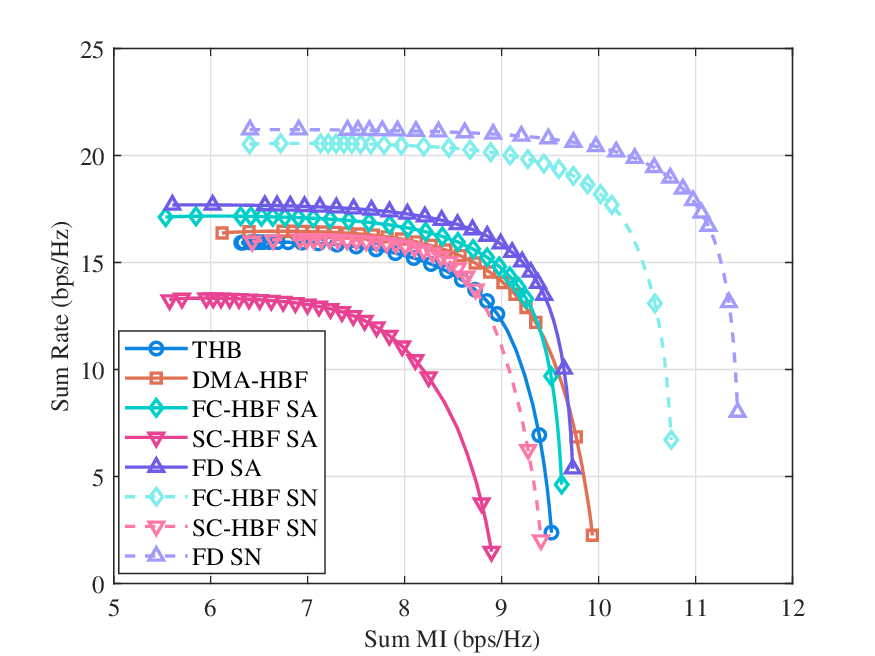}
}\hspace{-0.5 cm}
\subfigure[M=2.]{\label{fig5b:SR_vs_SMI_K4M2 }
		\includegraphics[width=0.32\linewidth]{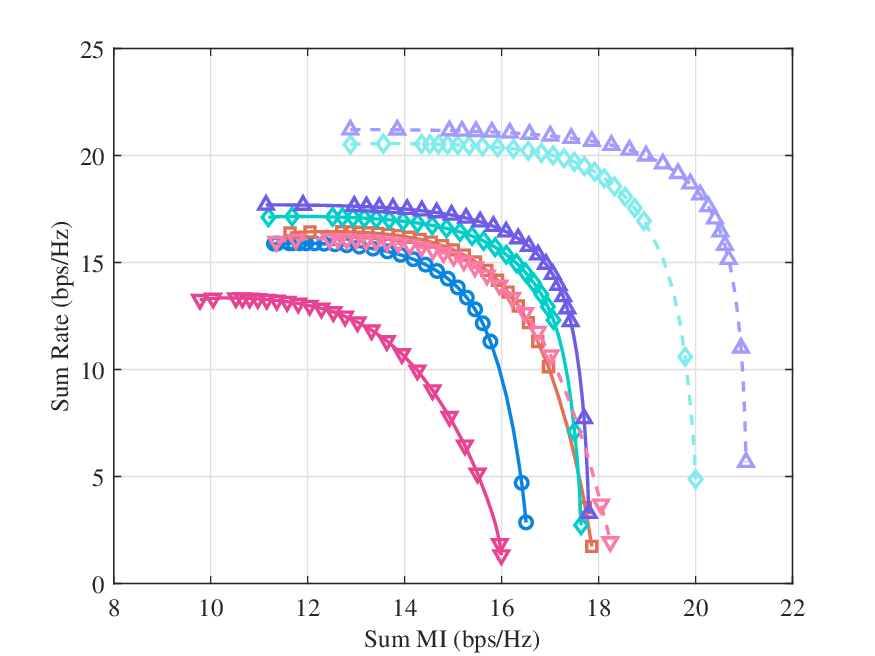}
}\hspace{-0.5 cm}
\subfigure[M=3.]{\label{fig5c:SR_vs_SMI_K4M3}
		\includegraphics[width=0.32\linewidth]{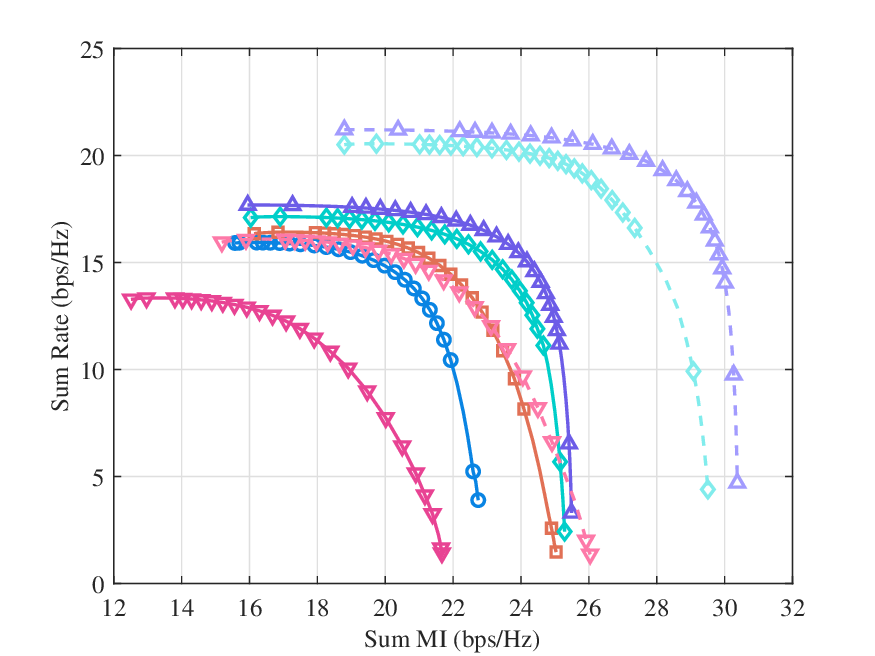}
}
	\caption{Comparison of different transceiver architectures with varying numbers of sensing targets ($M=1,2,3$) while fixing $K=4$. The tradeoff between sum rate and sum MI is illustrated for (a) $M=1$, (b) $M=2$, and (c) $M=3$.}

	\label{fig:performance_vs_M}
\end{figure*}
Fig.~\ref{fig:performance_vs_M} illustrates the performance of different transceiver architectures as the number of sensing targets increases from 1 to 3, with the number of communications users fixed at $K=4$. The results show that increasing the number of targets enlarges the sensing-only region, while the communications-only region remains nearly unchanged. We observe that the proposed THB architecture achieves a larger tradeoff region than the SC-HBF SA scheme, but a smaller one compared to the other baseline architectures. Moreover, the performance gap in the sensing region between the THB scheme and the other baselines widens as $M$ increases, which can be attributed to the aforementioned physical constraints of the THB architecture.

\section{Conclusion}\label{Sec:concu}
We have investigated a tri-HBF architecture for monostatic large-scale MIMO ISAC systems, where digital beamforming, analog phase shifters, and DMA-enabled electromagnetic processing are jointly optimized. By formulating a weighted sum maximization of communications sum rate and sensing MI under practical hardware constraints, we developed a low-complexity FP-based BCD algorithm with closed-form updates to efficiently tackle the resulting non-convex problem. Simulation results demonstrated that the proposed tri-HBF architecture achieves superior EE compared with fully digital, conventional hybrid, and DMA-based designs across the entire communications–sensing tradeoff regime. Although slight performance losses in sum rate and sensing MI arise due to inherent DMA hardware constraints, the overall results confirm the effectiveness of the tri-HBF architecture as a promising and energy-efficient solution for large-scale MIMO ISAC systems.

\appendix
\section{Appendix}
\numberwithin{lemma}{subsection} 
\numberwithin{corollary}{subsection} 
\numberwithin{remark}{subsection} 
\numberwithin{equation}{subsection}	

\subsection{Proof of Proposition \ref{pro:full power}}\label{property:full power}
  We prove this proposition by contradiction. For clarity, we define the total received signal power at user $k$ as $T_k\triangleq S_k+I_k+\sigma_{k}^2$, where $S_k\triangleq |\mathbf h_k^\H\mathbf W_\rme\mathbf W_\rma\mathbf w_{\rmd k}|^2, I_k\triangleq \sum_{j=1,j\neq k}^K |\mathbf h_k^\H\mathbf W_\rme\mathbf W_\rma\mathbf w_{\rmd j}|^2$. Suppose a locally optimal solution $\{\mathbf W_{\rme}^\diamond,\mathbf W_{\rma}^\diamond,\mathbf W_{\rmd}^\diamond\}$ exists such that the total transmit power satisfies $\|\mathbf W_{\rme}^\diamond\mathbf W_{\rma}^\diamond\mathbf W_\rmd^\diamond\|_\rmf^2<P_\rmt$.
Let $\varpi\in\mathbb R$ be a constant such that $ \sqrt {P_{\mathrm t}}/\|\mathbf W_{\rme}^\diamond\mathbf W_{\rma}^\diamond\mathbf W_\rmd^\diamond\|_\rmf\geq\varpi>1$, and define a scaled digital beamforming matrix $\mathbf W^\star\triangleq \varpi \mathbf W^\diamond$. The SINR for user $k$ then satisfies:
    \begin{align*}
    \frac{S_{k}(\mathbf W_{\rme}^\diamond,\mathbf W_{\rma}^\diamond,\mathbf W_{\rmd}^\diamond)}{I_k(\mathbf W_{\rme}^\diamond,\mathbf W_{\rma}^\diamond,\mathbf W_{\rmd}^\diamond)+\sigma_k^2}&< \frac{\varpi^2 S_k(\mathbf W_{\rme}^\diamond,\mathbf W_{\rma}^\diamond,\mathbf W_\rmd^\diamond)}{\varpi^2 I_k(\mathbf W_{\rme}^\diamond,\mathbf W_{\rma}^\diamond,\mathbf W_\rmd^\diamond)+\sigma_k^2}\\
    &=\frac{ S_k(\mathbf W_{\rme}^\diamond,\mathbf W_{\rma}^\diamond,\mathbf W_\rmd^\star)}{I_k(\mathbf W_{\rme}^\diamond,\mathbf W_{\rma}^\diamond,\mathbf W_\rmd^\star)+\sigma_k^2}.
\end{align*}
    The scaled $\mathbf W_\rmd^\star$ yields a strictly higher communications sum rate than $\mathbf W_\rmd^\diamond$, as it increases monotonically with SINR. For sensing MI, SCNR shares the same structure with SINR, indicating an improved sensing performance at point $\mathbf W_\rmd^\star$. The detailed derivations are omitted here for brevity. Overall, $\mathbf W_\rmd^\star$ achieves a strictly higher overall objective than $\mathbf W_\rmd^\diamond$, contradicting the local optimality of $\mathbf W_\rmd^\diamond$. We conclude that any locally optimal point, the power constraint must be active, i.e., satisfied with equality.

\subsection{Proof of Proposition \ref{pro:equivalence}}\label{App:equivalence}
We establish the equivalence between problem \eqref{P1} and its reformulation \eqref{P2} by leveraging the scale-invariance property of the objective function in \eqref{P2}, akin to the equivalence between Rayleigh quotient maximization and eigenvalue problems.

Let usa denote the objective functions of problem \eqref{P2} and \eqref{P1full} as $f_{1}(\bm\Upsilon,\mathbf W_\rmd)$ and $f_2(\bm\Upsilon,\mathbf W_\rmd)$, where $\bm\Upsilon\triangleq \{\mathbf W_{\rme},\mathbf P_{\rme},\mathbf W_{\rma},\mathbf P_{\rma},\mathbf P_\rmd\}$. Note that $f_1(\bm\Upsilon,\mathbf W)$ is scale-invariant, i.e., 
\begin{equation}
    f_1(\bm\Upsilon,\mathbf W_\rmd)=f_1(\bm\Upsilon,\varpi \mathbf W_\rmd), \quad \forall \varpi \in \mathbb{R} \setminus {0}.
\end{equation}
Using this property, we can write
\begin{align}
   \label{P2app} &\max_{\mathbf W_\rmd,\bm\Upsilon} f_1(\bm\Upsilon,\mathbf W_\rmd) \quad \text{s.t.} \eqref{P1C1},\eqref{P1C3}\\
   \label{P2app2}  &\overset{(a)}{\equiv} \max_{\mathbf W_\rmd,\bm\Upsilon} f_1(\bm\Upsilon,\mathbf W_\rmd) \quad \text{s.t.} \eqref{P1fullC1},\eqref{P1C1},\eqref{P1C3}\\
   \label{P2app3} &\overset{(b)}{\equiv} \max_{\mathbf W_\rmd,\bm\Upsilon} f_2(\bm\Upsilon,\mathbf W_\rmd) \quad \text{s.t.} \eqref{P1fullC1},\eqref{P1C1},\eqref{P1C3},
\end{align}
where step $(a)$ holds because we can always scale the beamforming matrix $\mathbf W_\rmd$ to satisfy the power constraint in \eqref{P1fullC1}. Specifically, define a dummy variable $\widetilde{\mathbf W}_\rmd=\sqrt{\frac{P_\rmt}{\|\mathbf W_{\rme}\mathbf W_{\rma}\mathbf W_\rmd\|_\rmf}}\mathbf W_\rmd$, and substitute $\widetilde{\mathbf W}_\rmd$ into problem \eqref{P2}. The scale invariance ensures the objective remains unchanged, while the add power constraint \eqref{P1fullC1} becomes valid. Step $(b)$ follows by explicitly enforcing the power constraint, thereby aligning the objective with that of problem \eqref{P1full}.

This establishes the equivalence between problems \eqref{P2} and \eqref{P1full}. Since problem \eqref{P1full} has already been shown to be equivalent to the original problem \eqref{P1} in proposition \ref{pro:full power}, the proof is complete.

\bibliographystyle{IEEEbib} 
\bibliography{IEEEabrv,reference}
%
%

\end{document}